\documentclass[acmsmall]{acmart}

\usepackage{amsmath,amsthm}
\usepackage{graphicx}
\usepackage[normalem]{ulem}
\usepackage{enumitem}
\usepackage{tikz}
\usetikzlibrary{positioning}
\usepackage{lipsum}
\usepackage{comment}
\usepackage[ruled,vlined,linesnumbered]{algorithm2e}
\usepackage{mathtools}
\usepackage{fontawesome}
\usepackage{array}
\usepackage{booktabs}

\usepackage{multirow}
\usepackage{makecell}
\usepackage{multicol}
\newcommand{\om}{CTL}

\newcommand{\otree}{\preceq}
\newcommand{\anc}{anc}

\newcommand{\Tau}{\mathcal{T}}

\newcommand*{\sep}{;\hspace{0.5em}}

\newcolumntype{H}{>{\setbox0=\hbox\bgroup}c<{\egroup}@{}}

\newtheorem{property}{Property}%

\renewcommand\footnotetextcopyrightpermission[1]{}

\AtBeginDocument{%
  }

\setcopyright{acmlicensed}
\copyrightyear{2018}
\acmYear{2018}
\acmDOI{XXXXXXX.XXXXXXX}

\begin{document}

\makeatletter
\def\ps@firstpagestyle{%
  \let\@oddhead\@empty
  \let\@evenhead\@empty
  \let\@oddfoot\@empty
  \let\@evenfoot\@empty
}
\makeatother

\pagestyle{plain}

\title{Optimized Customizable Route Planning in Large Road Networks with Batch Processing}


\author{Muhammad Farhan}
\affiliation{%
  \institution{Australian National University}
  \city{Canberra}
  \country{Australia}}
\email{muhammad.farhan@anu.edu.au}

\author{Henning Koehler}
\affiliation{%
  \institution{Massey University}
  \city{Palmerston North}
  \country{New Zealand}}
\email{h.koehler@massey.ac.nz}

\begin{abstract}
Modern route planners such as Google Maps and Apple Maps serve millions of users worldwide, optmizing routes in large-scale road networks where fast responses are required under diverse cost metrics including travel time, fuel consumption, and toll costs. Classical algorithms like Dijkstra or A$^*$ are too slow at this scale, and while index-based techniques achieve fast queries, they are often tied to fixed metrics, making them unsuitable for dynamic conditions or user-specific metrics. Customizable approaches address this limitation by separating metric-independent preprocessing and metric-dependent customization, but they remain limited by slower query performance.
Notably, Customizable Tree Labeling (CTL) was recently introduced as a promising framework that combines tree labelings with shortcut graphs. The shortcut graph enables efficient customization to different cost metrics, while tree labeling, supported by path arrays, provides fast query answering. Although CTL enables optimizing routes under different cost metrics, it still faces challenges in storing and reconstructing path information efficiently, which hinders its scalability for answering millions of queries. In this article, we build on the Customizable Tree Labeling framework to introduce new optimizations for the storage and reconstruction of path information. We develop several algorithmic variants that differ in the information retained within shortcut graphs and path arrays, offering a spectrum of trade-offs between memory usage and query performance. To further enhance scalability, we propose a batch processing strategy that shares path information across queries to eliminate redundant computation.
Empirically, we have evaluated the performance of our algorithms on 13 real-world road networks. The results show that they significantly outperform state-of-the-art methods, achieving speedups of up to factor 15 for route computation while maintaining practical memory requirements.
\end{abstract}

\maketitle

\section{Introduction}
Finding an optimal route between two locations in a road network is a fundamental task in modern navigation systems.
What constitutes an optimal route, however, can vary widely. Some users want the fastest path to minimize travel time, others prefer routes that minimize distance to reduce fuel consumption or aim to avoid tolls, and many applications must balance several factors at once. For example, logistics companies must minimize both delivery time and fuel costs while meeting customer deadlines~\cite{delling2017customizable}, eco-routing systems trade-off between travel time and vehicle energy consumption~\cite{ahn2008effects}, and multi-modal transport planners balance transfer times, ticket costs, and overall travel duration~\cite{bast2016route}. This diversity makes route planning more complex than simply finding the shortest path under a single fixed cost metric.

To optimize route computation at scale, a plethora of methods have been developed~\cite{geisberger2008contraction,goldberg2005computing,bast2006transit,arz2013transit,jung2002efficient,10.1007/11561071_51,chen2021p2h,ouyang2018hierarchy,abraham2012hierarchical,abraham2011hub,farhan2023hierarchical,akiba2014fast,jin2012highway,hart1968formal,sanders2006engineering,10.1145/2463676.2465277}, most notably Contraction Hierarchies (CH)~\cite{geisberger2008contraction}, Hub Labeling~(HL) \cite{abraham2011hub,abraham2012hierarchical}, and Transit Node Routing (TNR)~\cite{bast2006transit,arz2013transit}. These approaches preprocess the network to build auxiliary data structures that allow queries to be answered in milliseconds or even microseconds. Although highly effective when the cost metric is fixed, supporting multiple metrics would require building separate auxiliary structures, which can quickly becomes infeasible in terms of memory. Moreover, updating these structures when edge weights change, for example due to traffic conditions, is costly. While incremental maintenance techniques exist~\cite{ouyang2020efficient,zhang2022relative,farhan2025dual,koehler2025stable,geisberger2012exact}, they are only efficient when a very small number of changes in the metric occur at a time.

To address these challenges, customizable techniques were developed~\cite{delling2017customizable,dibbelt2016customizable,blum2022customizable,farhan2025customization}. Their key idea is to separate preprocessing into two phases. The first phase is independent of the metric and captures only the topological structure of the road network, which rarely changes. The second phase is metric-dependent and adapts this structure to a specific cost metric, which can then be used to answer queries efficiently. This approach makes it possible to efficiently reflect changing cost metrics without recomputing everything from scratch. The most well-known customization techniques are Customizable Route Planning (CRP) \cite{delling2017customizable} and Customizable Contraction Hierarchies (CCH) \cite{dibbelt2016customizable}. They achieve very fast customization but still suffer from slower query times compared to classical CH or TNR. Customizable Hub Labeling (CuHL) \cite{blum2022customizable} extended the customization idea to hub labeling, but this work remained theoretical and did not provide practical algorithms.

\renewcommand{\arraystretch}{1.2} 

\begin{figure}
    \centering
    \includegraphics[width=\textwidth]{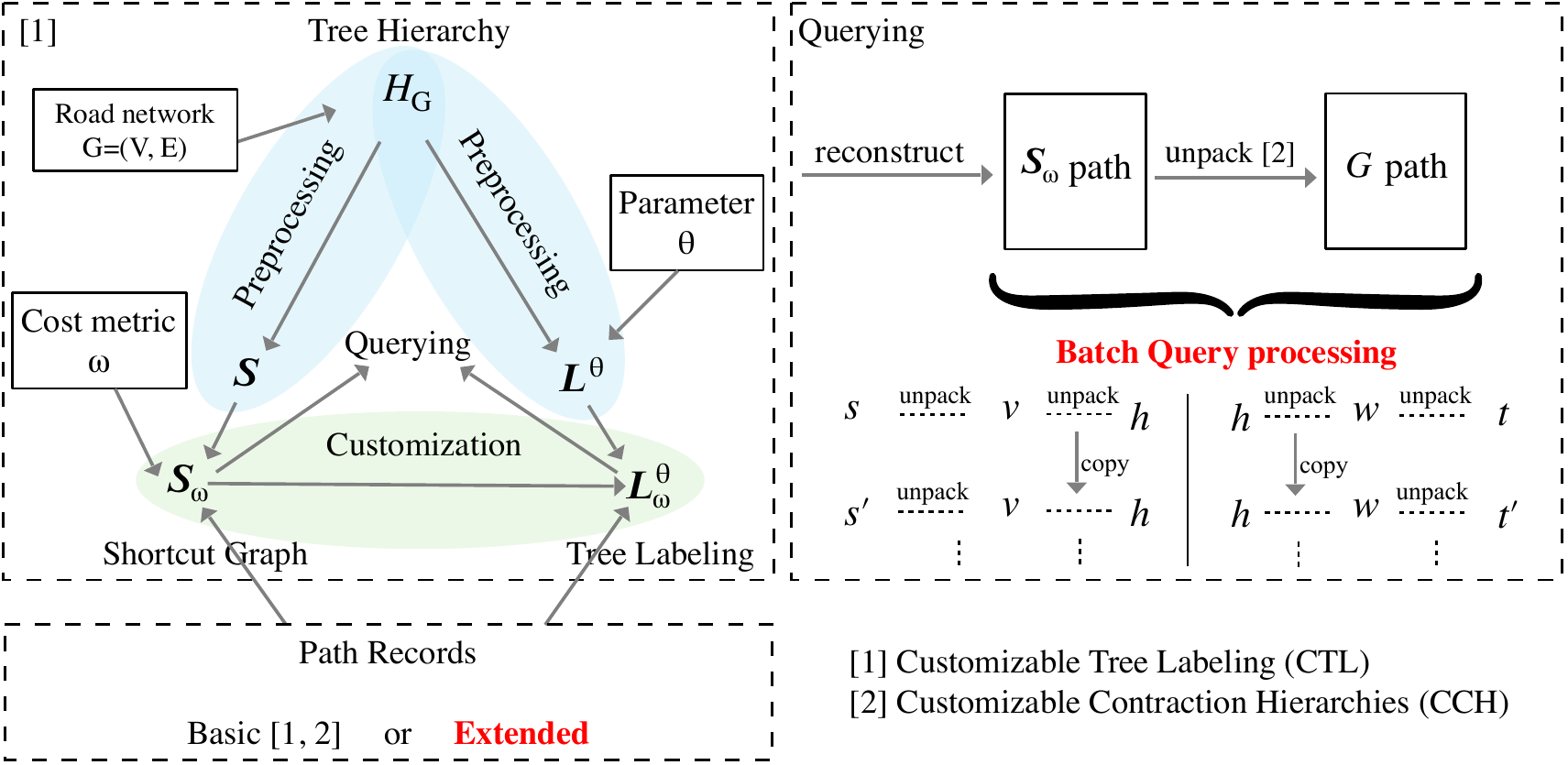}
    \caption{Framework overview: relationship to existing work; novel contributions highlighted in red.}
    \label{fig:overview}
\end{figure}

\begin{figure}
    \begin{minipage}{0.45\textwidth}
    \centering
    \resizebox{\textwidth}{!}{
    \begin{tabular}{cccc}
    \hline
    SG $\backslash$ PA & \textbf{n}one & \textbf{b}asic & \textbf{e}xtended \\
    \hline
    \textbf{b}asic    & CTL$_{bn}$ & CTL$_{bb}$ & -- \\
    \textbf{e}xtended & CTL$_{en}$ & CTL$_{eb}$ & CTL$_{ee}$ \\
    \hline
    \end{tabular}}\vspace{0.2cm}
    \end{minipage}
    \hfill
    \begin{minipage}{0.5\textwidth}
    \centering
    \includegraphics[width=\textwidth]{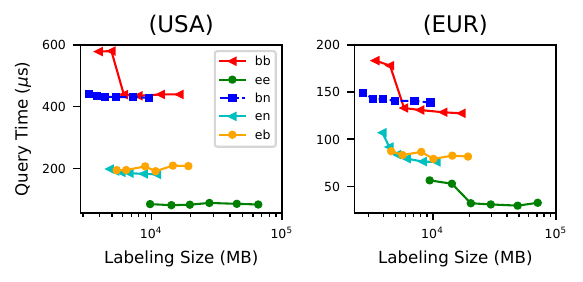}
    \end{minipage}
    \caption{(left) CTL variants from different storage choices in shortcut graphs (SG) and path arrays (PA). (right) Query time vs. labeling size for these variants on large networks.}
    \label{fig:variants}
    \label{fig:tradeoffs}
\end{figure}

Very recently, we proposed Customizable Tree Labeling (CTL) in \cite{farhan2025customization}, which combines ideas from labeling and hierarchical methods. CTL introduces a \emph{customizable labeling framework} based on the 2-hop cover property. The framework uses a tree hierarchy to capture the topological structure of the road network, a shortcut graph scheme to integrate new metrics, and tree labels to support query answering. CTL also introduces a parameterized customization method, giving flexibility to balance preprocessing effort, customization time, and query speed. While CTL represents an important advance, it also has limitations. In particular, the way CTL handles path information is not fully optimized for efficient route computation. Although CTL allows fast distance queries, reconstructing actual paths can be very slow, making it the main bottleneck.

In this article, we leverage CTL by exploring how different choices of precomputed path information affect performance. Our focus is on two components: the shortcut graph (SG) and the path arrays (PA). For SG, we consider basic storage, which only keeps the triangle node needed for path reconstruction, and extended storage, which records richer path information. For PA, we examine three options: no path arrays, basic storage that records only the endpoint of the first shortcut, and extended storage that includes both the endpoint and additional path information. Combining these options yields six variants of CTL, five of which are practical, as shown in Figure~\ref{fig:variants} (left). Each variant represents a different trade-off between memory requirements and query performance. Variants without path arrays use less memory but are slower at query time. In contrast, variants with extended path storage achieve much faster queries but at the cost of larger label sizes. Our experimental results shown in Figure~\ref{fig:tradeoffs} (right) confirm these trade-offs. In addition, we propose a batch processing approach which further improves query times when doing so is most critical, i.e., during high workloads.
Here we find that our 2-hop structure based on a shortcut graph naturally lends itself for identifying common subpaths which can be reused, as illustrated in Figure~\ref{fig:overview}.

\vspace{0.15cm}
\noindent\textbf{Contributions.~}This article is an extended version of~\cite{farhan2025customization}. The following contributions are novel, as highlighted in Figure~\ref{fig:overview}.
\begin{itemize}
    \item We propose an effective approach for storing additional path information within shortcuts and path arrays to speed up answering shortest path queries (Section~\ref{section:routing}).

    \item We introduce a batch processing approach that exploits shared subpaths among multiple query pairs. This reduces redundant computation, leading to significant improvements in scalability under high workloads (Section~\ref{sec:batch_processing}).
\end{itemize}
We evaluate the tradeoffs that storing additional path information brings theoretically and experimentally, and compare them against similar trade-offs stemming from parameterized labeling.
Results on 13 large real-world road networks show that our methods achieve significant speedups over state-of-the-art customizable approaches while keeping memory usage practical. In particular, the use of extended path information can reduce query times by factor 5 or more, while batch processing can speed up queries by a full order of magnitude.

\vspace{0.15cm}
\noindent\textbf{Outline.~}The rest of this article is organized as follows. Section~\ref{section:background} reviews related work in the literature. Section~\ref{section:preliminaries} introduces the basic notations and definitions used throughout the article. Section~\ref{section:framework} presents an overview of the framework. Section~\ref{section:tree-labeling} describes Customizable Tree Labeling (CTL), focusing on metric-independent preprocessing and metric-dependent customization. Section~\ref{section:parameterized} introduces parameterized tree labeling and the integrated querying algorithm that combines tree labeling with shortcut graphs. Section~\ref{section:routing} details our proposed algorithms for shortest path identification and their complexity analysis. Section~\ref{sec:batch_processing} introduced an efficient approach for processing shortest path queries in batches. Section~\ref{section:variants} discusses extensions of our work, including parallel customization and directed road networks. Section~\ref{section:experiments} reports our experimental results on large real-world road networks. Finally, Section~\ref{section:conclusion} concludes the article.

\newcommand{\abs}[1]{|#1|}
\section{Related Work}\label{section:background}

\subsection{Methods for Classical Routing}
Classical routing problem requires to find the minimum cost path considering a single static cost metric. Traditionally, Dijkstra’s algorithm~\cite{tarjan1983data} is used to solve classical routing problem. However, it may take several seconds to answer a single query on large road networks, which is impractical for applications which require to compute routes in the order of microseconds or nanoseconds. To accelerate route computation, numerous methods have been developed~\cite{geisberger2008contraction,goldberg2005computing,bast2006transit,arz2013transit,jung2002efficient,10.1007/11561071_51,chen2021p2h,ouyang2018hierarchy,abraham2012hierarchical,abraham2011hub,farhan2023hierarchical,akiba2014fast,jin2012highway,hart1968formal,sanders2006engineering,10.1145/2463676.2465277}. which can be broadly classified into two categories: 1) \emph{search-based methods}~\cite{geisberger2008contraction,goldberg2005computing,bast2006transit,arz2013transit,jung2002efficient,10.1007/11561071_51,sanders2006engineering,10.1145/2463676.2465277}, and 2) \emph{labelling-based
methods}~\cite{chen2021p2h,ouyang2018hierarchy,abraham2012hierarchical,abraham2011hub,farhan2023hierarchical,akiba2014fast,jin2012highway}. Among search-based methods,
Contraction Hierarchy (CH) [13] has demonstrated outstanding performance in practice. The key idea behind CH is to contract vertices in a particular order, by introducing shortcuts among their neighbors to maintain distance information. These shortcuts significantly reduce the search space during query time, leading to faster query responses. Despite its efficiency in pruning the search space, CH may still require exploring many paths.

To address the limitations of search-based methods, labelling-based methods have been developed with great success~\cite{abraham2011hub,akiba2013fast,ouyang2018hierarchy,chen2021p2h,farhan2023hierarchical,akiba2014fast,bast2006transit,arz2013transit}. These methods precompute labels for all vertices that capture the shortest path information. Rather than performing a search over the graph, the algorithm simply examines the precomputed labels to retrieve the path information. Labelling-based methods can find routes significantly faster than search-based methods, at the cost of requiring additional space for storing labels. Notably, the most advanced labelling-based approaches~\cite{ouyang2018hierarchy,chen2021p2h,farhan2018highly} exploit hierarchical structures of road networks to reduce the search space on the labels at query time. Hierarchical 2-Hop (H2H)~\cite{ouyang2018hierarchy,ouyang2023hierarchy} and Projected vertex separator based 2-Hop labeling (P2H)~\cite{chen2021p2h} utilize the tree decomposition~\cite{bodlaender2006treewidth} of a road network to define a vertex hierarchy. Hierarchical Cut 2-hop Labelling (HC2L) \cite{farhan2023hierarchical} recursively partitions a road network to construct a balanced tree hierarchy among vertices. This balanced tree hierarchy enabled HC2L to further reduce the search space on labels at query time, and makes it the current state-of-the-art method for classical routing in static road networks.
Dual Hierarchy Labeling (DHL)~\cite{farhan2025dual} extends this approach to dynamic networks.

\subsection{Methods for Customizable Routing}
Dynamic changes on edge weights in road networks, such as those caused by varying traffic conditions or road closures, necessitate customizable routing that can efficiently update precomputed data for accurate querying. Existing methods for customizable routing broadly fall into two categories: 1) \emph{incremental maintenance methods}~\cite{geisberger2012exact,zhang2022relative,zhang2021dynamic,ouyang2020efficient,farhan2025dual,koehler2025stable} and 2) \emph{customizable methods}~\cite{dibbelt2016customizable,delling2017customizable,blum2022customizable}.
These are also referred to as partially or fully customizable methods, respectively.

Incremental maintenance methods identify and repair parts of the precomputed auxiliary data structure affected by dynamic updates, rather than recomputing everything from scratch. However, this approach is only effective when the fraction of edge weights changed is very small, and becomes impractical in scenarios where large parts of the network are affected, e.g. reported travel times which may get updated every few minutes.

Customizable methods separate topological structure from metric properties. This division allows the metric-independent preprocessing phase to focus solely on the network structure, enabling quick adjustments to varying metrics without extensive re-computation. The well-known customisation techniques are customisable route planning (CRP) \cite{delling2017customizable} and customizable contraction hierarchies (CCH) \cite{dibbelt2016customizable}. CRP \cite{delling2017customizable} pre-computes multilevel partition-based overlay graphs using separa\-tor-based techniques. CCH \cite{dibbelt2016customizable} utilize a nested dissection order \cite{george1973nested} for constructing a CH, and customize the CH edges for query answering. Both methods utilize metric-independent auxiliary data to avoid re-computa\-tion of the entire process, making them highly practical for scenarios where edge weights frequently change. However, their low customization times and small memory footprints are offset by query times that are significantly higher than those achieved by labelling-based methods. Another very recent work, Customizable Hub Labelling (CHL) \cite{blum2022customizable}, exploits theoretical properties of CCH to apply the customization paradigm to hub labelling (HL)~\cite{abraham2011hub,abraham2012hierarchical}. While it has state-of-the-art query performance, customization is very inefficient.
\section{Preliminaries}\label{section:preliminaries}
Let \( G = (V, E) \) represent a road network. A \emph{(cost) metric} is a function \( \omega: E \to \mathbb{R}_{>0} \) that assigns a positive cost \( \omega(u, v) \) to each edge \( (u, v) \in E \), such as travel time. This metric may be initially unknown or dynamically determined based on the road network's properties or external factors. A path is a sequence of distinct vertices \( p = (v_1, v_2, \dots, v_k) \), where \( (v_i, v_{i+1}) \in E \) for all \( 1 \leq i \leq k-1 \). The cost of a path \( p \), given a metric \( \omega \), is 
$\omega(p) = \sum_{i=1}^{k-1} \omega(v_i, v_{i+1}).$ An \emph{optimal route} between two vertices \( s \) and \( t \) minimizes the total path cost with respect to a metric \( \omega \). We denote the set of all optimal routes between \( s \) and \( t \) in \( G \) under the metric \( \omega \) by \( P^{\omega}_G(s,t) \), and the \emph{optimal cost} by \( d^\omega_G(s, t) \). 

We formally define the customizable routing problem as follows.
\begin{definition}[Customizable Routing]  
In a road network \( G = (V, E) \), for any two vertices \( s, t \in V \), the \emph{customziable routing problem} is to efficiently find an optimal route \( p \in P^{\omega}_G(s, t) \) from \( s \) to \( t \) with respect to any given cost metric $\omega$.
\end{definition}
To compute an optimal route from one vertex to another, finding the optimal cost is crucial as it restricts the search space and ensures correctness during path restriction. Customizable routing also applies to both directed and undirected graphs.

\begin{figure}[t]
    \centering
    \includegraphics[width=0.37\textwidth]{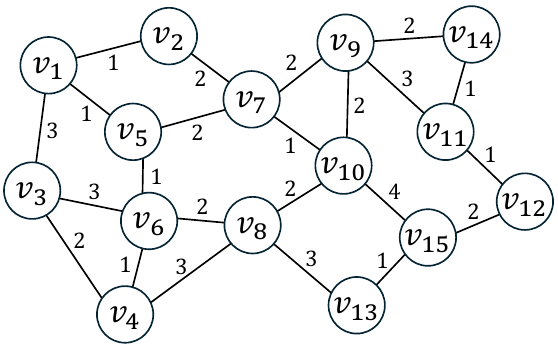}
    \caption{A road network $G$, customized with a cost metric $\omega$.} \label{fig:example_road_network}
\end{figure}

\begin{example}
Figure~\ref{fig:example_road_network} shows a road network $G = (V, E)$ with
15 vertices and 23 edges. Each $e\in E$ is assigned a cost by a metric $\omega$. Consider two vertices, \( v_1 \) and \( v_{12} \), connected by multiple paths. For instance, 
\( p_1 = \langle v_1, v_2, v_7, v_9, v_{11}, v_{12} \rangle \) with cost \( \omega(p_1) = 9 \), 
and \( p_2 = \langle v_1, v_5, v_7, v_{10}, v_{15}, v_{12} \rangle \) with cost \( \omega(p_2) = 10 \). 
Among these, \( p_1 \) is optimal as it has the minimum cost.
\end{example}

Customizable approaches are ineffective when each query uses a unique metric, as the cost of customization typically outweighs that of an index-free search. However, customization works when metrics are frequently reused, such as when the number of distinct metrics is small and multiple customizations can be stored in memory, or when updates reflect periodic changes in costs like travel time.
\section{Framework Overview}\label{section:framework} 
This section introduces a customizable labeling framework designed to enhance query performance. The main challenge lies in designing a preprocessing strategy that generalizes across diverse metrics while achieving efficiency in both customization and querying. To address this, we introduce a parameterization process for labeling that balances customization and query efficiency.




\begin{definition}[2-Hop Labeling \cite{cohen2003reachability}]
Given a cost metric $\omega$, a \emph{2-hop labeling} over $G = (V, E)$ assigns each vertex $v \in V$ a label $L(v)$ containing cost entries $\{(u_1,\delta_{vu_1}), \dots, (u_k,\delta_{vu_k})\}$, where $u_i \in V$.
These labels must ensure that
\[
d^{\omega}_G(s, t) = \min_{v \in L(s) \cap L(t)} \big\{\delta_{sv} + \delta_{tv}\big\}
\]
holds for all $s,t\in V$.
\end{definition}

We shall also refer to a 2-hop labeling w.r.t. $\omega$ as an \emph{$\omega$-labeling} to emphasize the particular metric used.
Note that typically $\delta_{vu_i} = d^{\omega}_G(v, u_i)$, though some approaches use distances within subgraphs instead.
The vertex identifiers $u_i$ are omitted at times, with array indices taking their role in matching cost values.





\subsection{Parameterization}
While 2-hop labelings are highly efficient for query answering, they tend to suffer from large label sizes, compared to shortcut based approaches like CCH or CRP.
We therefore will employ only a \emph{partial} 2-hop cover which only maintains part of the label structure based on a parameter $\theta$.
At one extreme ($\theta=0$) we obtain a full 2-hop cover (specifically a hierarchical customizable hub labeling in the terminology of \cite{blum2022customizable}), at the other ($\theta=\infty$) we obtain a customizable contraction hierarchy, while any parameter setting in between provides a tradeoff between these.

Achieving favorable trade-offs can be challenging, particularly when integrating techniques with differing optimization objectives. For instance, labeling techniques might achieve a query time of \(1 \mu s\) with a label size of \(100 \, \text{GB}\), while shortcut-based methods might provide a \(1 \, \text{ms}\) query time with a label size of \(100 \, \text{MB}\). A combined approach averaging \(500 \, \mu s\) query time and \(50 \, \text{GB}\) label size would often be deemed both slow and large, representing the worst of both worlds.

\subsection{Customizable Labeling Framework}
Let $\mathcal{G}$ represent the set of all road networks, 
$W$ the set of possible cost metrics, $\mathcal{D}$ the set of all data structures independent of any metric, and $\mathcal{D}_{\omega}$ the set of all data structures dependent of a metric $\omega$. 
Our framework consists of three key components:

\begin{itemize}
    \item \emph{\underline{Preprocessing algorithm:}} A function $\mathcal{A}_P: \mathcal{G} \times \mathbb{R}\to \mathcal{D}$ that maps a road network $G$ and parameter $\theta$ to a metric-independent data structure
    $\mathcal{A}_P(G)=(L^{\theta}, S)$. Here, $L^{\theta}$ is a parameterized labeling and satisfies the customizable cover property~\cite{blum2022customizable} when $\theta=0$, and $S$ is a data structure to accelerate customization.
    \item \emph{\underline{Customization algorithm:}} A function
    $\mathcal{A}_C: \mathcal{D} \times W \to \mathcal{D}_{\omega}$ that customizes a preprocessed data structure $\mathcal{A}_P(G) \in \mathcal{D}$ with respect to a metric $\omega \in W$.
    That is, \( \mathcal{A}_C(\mathcal{A}_P(G), \omega) = (L^{\theta}_{\omega}, S_{\omega}) \), where \( L^{\theta}_{\omega} \) and $S_{\omega}$ are customized from $L^{\theta}$ and $S$, respectively.
    By the customizable cover property, \( L^{\theta}_{\omega} \) is guaranteed to be a 2-hop \( \omega \)-labeling when $\theta=0$. 
    \item \emph{\underline{Query algorithm:}} A function
    $\mathcal{A}_Q: \mathcal{D}_{\omega} \times V \times V \to R$ that takes a customized data structure $D_{\omega}=(L^{\theta}_{\omega}, S_{\omega})$ and two vertices $s, t \in V$, returning either the optimal cost or an optimal route from $s$ to $t$ with respect to the metric $\omega$.
\end{itemize}

While conceptually straightforward, designing such a metric-independent data structure while balancing customization efficiency and query performance across varying metrics remains a challenging task. In particular, metric-independent data structure must preserve topological structure of road networks without relying on fixed metrics; customization must efficiently handle large-scale updates; and queries must remain fast despite varying label size.
In the following we first present a basic non-parameterized approach which offers fast query times but suffers from large labeling size and customization times (Section~\ref{section:tree-labeling}).
This then forms the basis for our parameterized approach, which balances query time and labeling size more favorably (Section~\ref{section:parameterized}).


\section{Tree-Based Customizable Routing}\label{section:tree-labeling}
This section presents a tree-based solution for customizable routing, enabling efficient preprocessing, customization, and querying. A key challenge in designing tree-based customizable routing solution is to support efficient topology-based preprocessing that yields compact and query-friendly data structures. Traditional labeling-based methods suffer from large label sizes~\cite{blum2022customizable,abraham2011hub,abraham2012hierarchical}, while shortcut-based methods~\cite{dibbelt2016customizable,delling2017customizable} favor customization speed at the cost of slower queries. Our initial tree-based approach presented here falls into the former category, but also maintains an intermediate short-cut based structure which is used during customization.

\subsection{Metric-Independent Preprocessing}

\subsubsection{Tree Hierarchy}
To efficiently support customizable routing, we introduce a \emph{tree hierarchy}, which provides a topology-based representation of road networks.

\begin{definition}[Tree Hierarchy]\label{def:td}  
Let \(\beta \in (0, 0.5)\). A \emph{tree hierarchy} over a road network \(G=(V,E)\) is a binary tree \(H_G = (\mathcal{N}, \mathcal{E}, f)\), where \(\mathcal{N}\) is the set of tree nodes, \(\mathcal{E}\) is the set of tree edges, and \(f: V \rightarrow \mathcal{N}\) is a total surjective function satisfying two conditions:  

\begin{enumerate}[leftmargin=15pt]   
    \item \emph{\underline{Balanced subtrees:}} For every internal node \(N \in \mathcal{N}\), the left and right subtrees satisfy:
    \[
    |T_{\ell}(N)|, |T_r(N)| \leq (1 - \beta) \cdot |T_{\ell}(N) \cup T_r(N)|,
    \]
    where \(T_{\ell}(N)\) and \(T_r(N)\) are the sets of vertices in the left and right subtrees of \(N\), respectively.

    \item \emph{\underline{Ancestor separation:}} For any two vertices \(s, t \in V\), the set
    \[
    CA(s, t) = \{v \in V \mid f(v) \in A(f(s)) \cap A(f(t))\}
    \]
    of their common tree ancestors contains at least one vertex on each path between \(s\) and \(t\) in \(G\), where \(\mathsf{A}(\cdot)\) denotes the set of ancestor nodes of a tree node.
 
\end{enumerate}  
\end{definition}

Figure~\ref{fig:labeling} shows a tree hierarchy for the road network in Figure~\ref{fig:example_road_network}, with the root containing vertices $\{v_7,v_8\}$.
The balanced subtrees condition keeps the tree compact, ensuring efficiency in label construction and query processing. The ancestor separation condition enables the use of common ancestors as hubs for 2-hop labeling.

Constructing such a tree hierarchy requires effective graph bipartitioning to minimize separator sizes while maintaining balance. Although this problem is NP-hard, recent heuristics have demonstrated excellent scalability and performance on large-scale road networks~\cite{delling2011graph,Schulz13a,farhan2023hierarchical}. In this work, we adopt the recursive bi-partitioning method from~\cite{farhan2023hierarchical,hu2025reproducibility}, which iteratively identifies balanced separators to partition the graph.

A hierarchical vertex order \( \otree \) can be derived from the tree hierarchy \( H_G \). Specifically, $\otree$ is a partial order on the vertices of the road network $G$: for any vertices \( u, v \in V \) with \(f(u)\neq f(v)\) we have \(v \otree u\) iff \(f(v) \in A(f(u))\). Vertices mapped to the same tree node are totally ordered by $\otree$, though this order is arbitrary.

We use $\anc(v)=\{ w\in V \mid w\otree v \}$ to denote the ancestor vertices of $v$.
The \emph{rank} of a vertex is defined as $\tau(v)=|\anc(v)|$. Although different vertices can have the same rank, the ranks of ancestor within $\anc(v)$ are distinct, running from \(1\) to \(\tau(v)\), and will be used as indices within labels. The rank of a tree node is the maximum rank of vertices mapped to it.

\begin{example}
Consider the tree hierarchy shown in Figure~\ref{fig:labeling}. By ordering the vertices within each tree node using their node identifiers, we obtain the hierarchical vertex order $\otree := \{7{\otree}8{\otree}10,3; 10{\otree}12{\otree}\\9,15; 3{\otree}5{\otree}1,6; 9{\otree}14{\otree}11; 15{\otree}13; 1{\otree}2; 6{\otree}4\}$. We also have \(\anc(v_3)=\{v_7,v_8,v_3\}\) and \(\anc(v_5)=\{v_7,v_8,v_3,v_5\}\).
Their ranks are \(\tau(v_3)=3, \tau(v_5)=4\) and the rank of tree node \(f(v_3)\) is \(max(3,4)\). 
\end{example}

\subsubsection{Tree Labeling Scheme}
Based on the tree hierarchy $H_G$, we define the following tree labeling scheme.

\begin{definition}[Tree Labeling Scheme]\label{def:HL}  
Let \(H_G\) be a tree hierarchy over a road network \(G = (V, E)\). A \emph{tree labeling scheme} is a tuple \(L=(H_G, \mathcal{I}, \Tau, C)\), where: 

\begin{itemize}[leftmargin=20pt]  
    \item \underline{\(\mathcal{I} = \{\mathcal{I}(v) \mid v \in V\}\):} Each vertex \(v \in V\) is assigned a bitstring identifying the position of node \(f(v)\) in the tree hierarchy.
    \item \underline{\(\Tau = \{\Tau(v)\mid v\in V\}\):}
    Each $\Tau(v)$ is a \emph{rank array}
    \([\tau(N_1),\dots,\tau(N_k),\tau(v)]\) where \(\{N_1,\dots,N_k,f(v)\}\) are the ancestor tree nodes $A(f(v))$ of $f(v)$.

    \item \underline{\({C} = \{{C}(v) \mid v \in V\}\):} Each \({C}(v)\) is a \emph{cost array}
    \([\delta_{vw_1}, \dots,\delta_{vw_k}]\), where
    \(\anc(v)=\{w_1, \dots, w_k\}\) and $\delta_{vw_i}$ indicates a cost value (not necessarily minimal) for paths between $v$ and $w_i$.
\end{itemize}  
\end{definition}

\begin{figure}[ht!]
    \centering
    \includegraphics[width=0.77\textwidth]{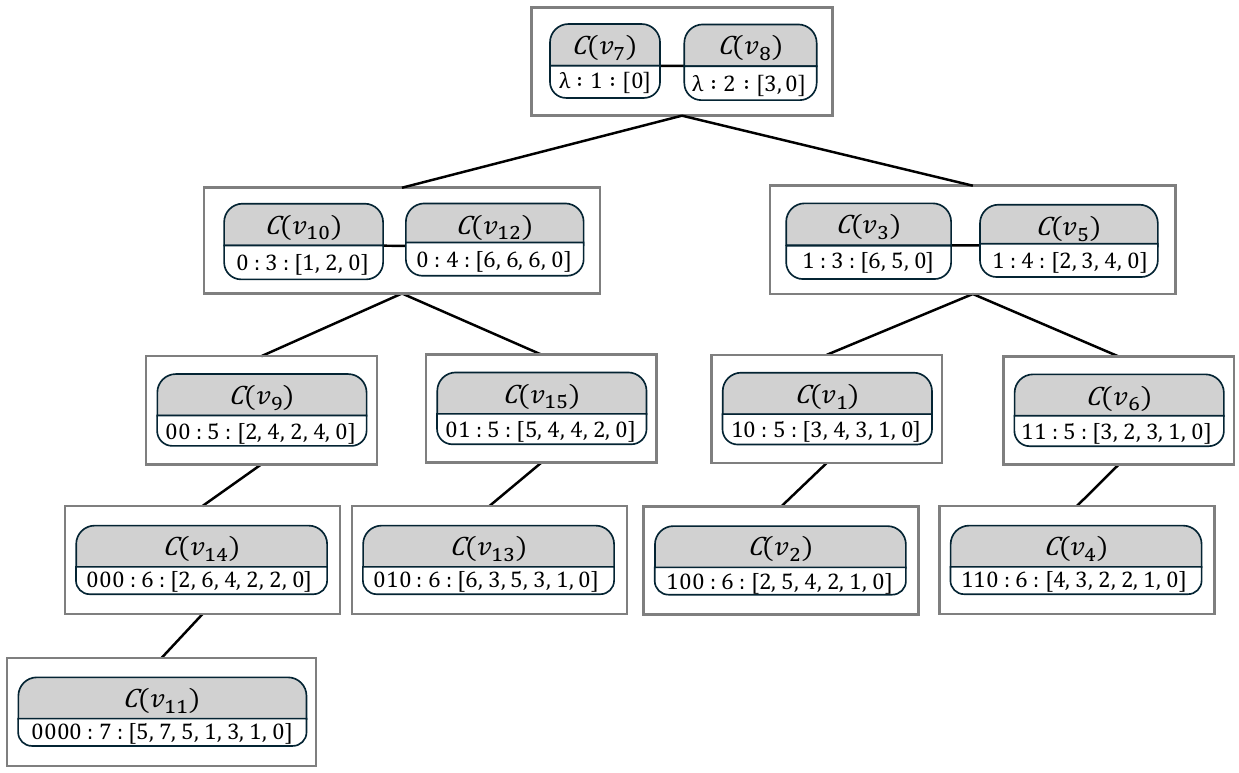}
    \caption{An illustration of a tree hierarchy $H_G$ and tree labeling Scheme \( L \) over $G$, where vertices at the first level are assigned the empty bitstring $\lambda$. We only show $\tau(v)$ in place of $\Tau(v)$ for brevity.}
    \label{fig:labeling}
\end{figure}

\begin{example}
Figure~\ref{fig:labeling} shows a tree hierarchy and tree labeling scheme over an example road network $G$ depicted in Figure~\ref{fig:example_road_network}. The label information for each vertex $v$ is presented under its label $L(v)$ in the format $\mathcal{I}(v) : \tau(v) : \mathcal{C}(v)$. 
Consider vertex $v_{15}$: it is assigned a node identifier $\mathcal{I}(v_{15})=01$ as its node is the right child of the left child of the root. Its ancestors are $\anc(v_{15})=\{v_7, v_8,\ v_{10}, v_{12},\ v_{15}\}$, resulting in a rank array \(\Tau(v_{15})=[2,4,5]\) and a cost array $[5, 4, 4, 2, 0]$ storing distances to its ancestors $\{v_7, v_8, v_{10}, v_{12}, v_{15}\}$.
\end{example}

The tree labeling scheme leverages the structure of the tree hierarchy \(H_G\) to assign labels that encode both topological and cost-related information for efficient query processing. The node identifiers (\(\mathcal{I}\)) together with the rank arrays (\(\Tau\)) enable efficient computation of how many common ancestors two vertices have, and thus which prefix of the cost arrays (\(C\)) to use as hub for 2-hop cost computation.
In fact, \(\mathcal{I}\) and values in \(\Tau\) other than \(\tau(v)\) are only used by \texttt{GetLcaHeight} in Algorithm~\ref{algo:query}.

\subsubsection{Shortcut Graph Scheme}
Fast customization requires efficient path searches. To support this, we utilize a data structure called \emph{shortcut graph scheme}. With the hierarchical vertex order $\otree$ we extend the road network by adding \emph{shortcuts} between vertices where intermediate vertices have higher ranks than the endpoints.

\begin{definition}[Shortcut Graph Scheme]  
Given a road network \( G = (V, E) \) and a tree hierarchy \( H_G \), the \emph{shortcut graph scheme} \( S = (V, E^*, \otree) \) consists of the vertex set $V$, the edge set \( E^* = E \cup E' \) that includes all edges $E$ in \( G \) and shortcuts \( E' \), and a hierarchical vertex order \( \otree \) induced by \( H_G \).  A shortcut \( (v, u) \in E' \) is added if \( (v, u) \notin E \) and there exists a \emph{valley path} \( p \) between \( v \) and \( u \), meaning that \( v, u \otree w \) for all \( w \in V(p) \setminus \{v, u\} \).
\end{definition}

For each vertex \( v \in V \), we denote its \emph{upward neighbors} which ranked \emph{lower} than $v$ as \(     N^+(v) = \{ u \mid (v,u) \in E^*\land u \otree v \} \), and \emph{downward neighbors} which are ranked \emph{higher} than \( v \) as \( N^-(v) = \{ u \mid  (v,u) \in E^* \land v \otree u \} \) in the shortcut graph scheme \( S \).

We construct  
$S$ following the method proposed in~\cite{ouyang2020efficient}, with two key modifications: (1) \emph{Hierarchical vertex  contraction:} Vertices are contracted based on the hierarchical vertex order $\otree$, starting with the highest-ranked vertices (in decreasing order of rank). (2) \emph{Shortcut addition:} Edge costs are not computed (initially).

\begin{figure}[ht!]
    \centering  \includegraphics[width=0.77\textwidth]{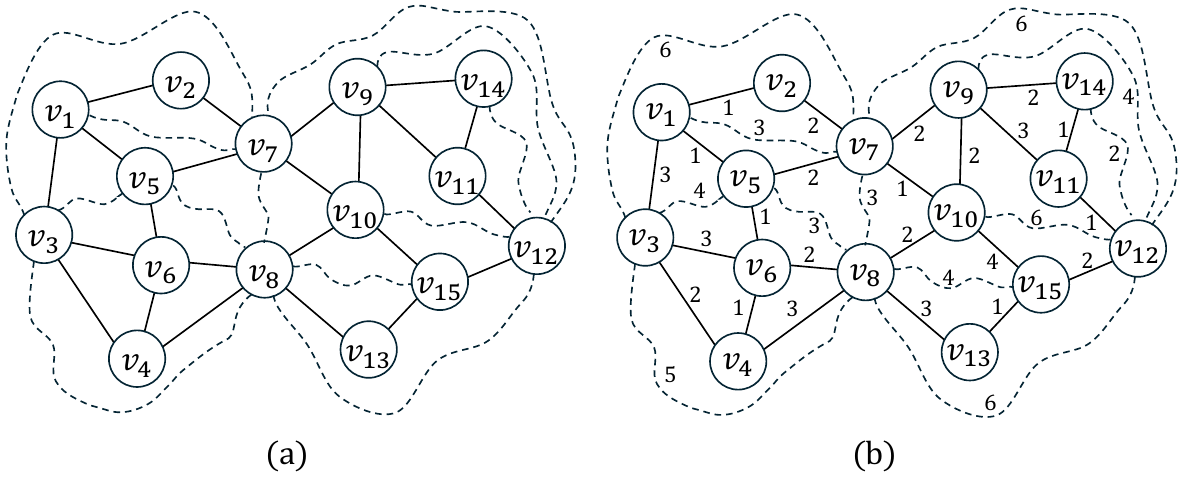}
    \caption{(a) An illustration of shortcut graph scheme; (b) Shortcut graph scheme after customization.}
    \label{fig:sc-graph}
\end{figure}
\begin{example}
Figure~\ref{fig:sc-graph}(a) shows a shortcut graph scheme $S$ over an example road network $G$ shown in Figure~\ref{fig:example_road_network} using the hierarchical vertex order
induced by the tree hierarchy $H_G$ depicted in Figure~\ref{fig:labeling}. Dashed edges represent shortcuts. A shortcut edge $(v_1,v_7)$ in $S$ is added because a valley path $\langle v_1, v_2, v_7\rangle$ between $v_1$ and $v_7$ exist in $G$. Note that $\langle v_1,v_5,v_7\rangle$ is not a valley path since $v_5\otree v_1$. Consider vertex $v_1$, we have \( N^+(v_1) = \{ v_3, v_5, v_7 \} \) and \( N^-(v_1) = \{ v_2 \} \).
\end{example}

\subsection{Metric Customization}

We describe the customization, which integrates a given cost metric 
 \(\omega\) into tree labeling scheme $L$ and shortcut graph scheme $S$ through two sequential steps:  
(1). Customize \( S \) to $S_{\omega}$ by incorporating the edge costs of the road network into all shortcuts.  
(2). Customize $L$ to $L_{\omega}$ by computing the cost array for each \( C(v) \) based on \( S_{\omega} \). 

\subsubsection{Customizing Shortcut Graph Scheme}
The \emph{shortcut customization property} is utilized to customize the costs of shortcuts in $S$.

\begin{property}[Shortcut Customization]\label{lab:weight_property}
Let $S=(V,E^*,\otree)$. For any $(v,u)\in E^*\setminus E$, the cost \(\omega(v,u)\) in \( S_{\omega} = (V, E^*, \omega, \otree)~\) satisfies: 
\begin{align*}\label{eq:weight_property}
\omega(v,u)=\min\{\omega(w, v)+\omega(w, u)\mid w\in N^-(v)\cap N^-(u) \}.
\end{align*} 
\end{property}

Specifically, the shortcuts in \( S \) are customized in descending order of \( v \) with respect to \( \otree \), as described in Algorithm~\ref{algo:custom-ch}. For each upward neighbor \( u \in N^+(v) \), the algorithm iterates over the common downward neighbors \( w \) of \( v \) and \( u \). It checks if the path 
between $v$ and $u$ through $w$ in $S$ has a lower cost than the direct path between $v$ and $u$.
If so, it updates \(\omega(v, u)\) to reflect the lower cost between \( v \) and \( u \) in \( G \) (Lines 5–7), based on Property~\ref{lab:weight_property}. By construction, \( w \) has a strictly higher rank than \( v \) and \( u \), ensuring that the costs \(\omega(v, w)\) and \(\omega(w, u)\) are already computed before \(\omega(v, u)\) is updated. 

\begin{example}
Figure~\ref{fig:sc-graph}(b) shows shortcut graph scheme \( S \) after metric customization, with weights assigned based on the metric in Figure~\ref{fig:example_road_network}. Consider the following shortcuts in the order processed by Figure~\ref{algo:custom-ch}:
the cost of $(v_{12},v_{14})$ is computed as $\omega(v_{12},v_{14})=\omega(v_{12},v_{11})+\omega(v_{11},v_{14})=2$ using the triangle $\langle v_{12},v_{11},v_{14}\rangle$. This is then used to compute $\omega(v_{12},v_9)=\omega(v_{12},v_{14})+\omega(v_{14},v_{9})$ $=4$ and finally $\omega(v_{12},v_7)=\omega(v_{12},v_9)+\omega(v_9,v_7)=6$.
\end{example}

\begin{algorithm}[t]
\caption{Customizing Shortcut Graph}\label{algo:custom-ch}
\SetCommentSty{textit}
\SetKwFunction{FMain}{CustomizeS}
\SetKwProg{Fn}{Function}{}{end}
\SetKw{and}{and}
\Fn{\FMain{$G$, $S$}}{
    $S_\omega\gets S$
    
    \noindent\tcp{initialize new edge costs for $S_\omega$}
    \ForEach{$(u, v) \in E$}{
          $\omega(u, v)\gets\omega_G(u, v)$
    }

    \tcp{customize shortcut costs}
    \ForEach{$v \in V$ in descending order of $\otree$}{
        \ForEach{$u\in N^+(v)$}{
            
            \ForEach{$w\in N^-(v)\cap N^-(u)$}{
                $\omega(v, u)\gets\min\{\omega(v, u), \omega(w, v) + \omega(w, u)\}$\\
            }
        }
    }
}
\end{algorithm}
\begin{algorithm}[t]
\caption{Customizing Tree Labeling}\label{algo:label-construct}
\SetCommentSty{textit}
\SetKwFunction{FMain}{CustomizeL}
\SetKwProg{Fn}{Function}{}{end}
\SetKw{and}{and}
\Fn{\FMain{$S_\omega$}}{
    \tcp{initialize}
    \ForEach{$v\in V$}{
        $C(v)\gets[\infty,\dots,\infty]$, $\delta_{vw}\gets 0$, s.t., $w=\tau(v)$\\
    }
    \tcp{customize label distances}
    \ForEach{$v\in V$ in ascending order of $\otree$}{
        \ForEach{$u\in N^+(v)$}{
            \ForEach{$w$ s.t. $w\preceq u$\label{L:ch-construct-iter-anc}}{
                $\delta_{vw}\gets\min(\delta_{vw},\; \omega(v,u) + \delta_{uw})$\\
            }
        }
    }
}
\end{algorithm}

\subsubsection{Customizing Tree Labeling Scheme} 
To customize the cost array \([\delta_{vw_1}, \dots, \delta_{vw_k}]\) for each \( C(v) \), we apply the following \emph{label customization property}. This property ensures efficient computation of costs  using the hierarchical structure of $S_{\omega}$.

\begin{property}[Label Customization]\label{eq:label_property}
Let \( S_{\omega} = (V, E^*, \omega, \otree)~\). For any \( v, w \in V \) with \( w \otree v \), each cost entry in \( C(v) \) satisfies:  
\begin{align*}  
    \delta_{vw_i} = \min_{u \in N^+(v)} \{ \omega(v, u) + \delta_{uw_i} \}, \quad \text{for } 1 \leq i \leq k.  
\end{align*}  
\end{property}

We customize the cost array \( C(v) \) for each vertex \( v \in V \) in ascending order of \( \otree \) using a top-down approach, as detailed in Algorithm~\ref{algo:label-construct}. Each cost array \( C(v) \) is initialized to $\infty$, i.e., \( \delta_{vw} = \infty \) for all \( w \neq v \) and to self as $0$ (Lines 2-3). For each \( u \in N^+(v) \), we iterate over all $\{w\mid w\otree u\}$. If a path with lower cost is found, i.e., \( \delta_{vw} > \omega(v, u) + \delta_{uw} \), we update \( \delta_{vw} \) with \( \omega(v,u) + \delta_{uw} \). This operation leverages Property~\ref{eq:label_property}. A key observation is that when processing \( v \), the labels for all vertices $\{w\mid w\otree v\}$ have already been processed and thus can be used together with its upward neighbors to compute $C(v)$. 

\begin{example}
Consider Figures~\ref{fig:labeling} and~\ref{fig:sc-graph}(b). Suppose the costs for all vertices with ranks lower than 5, namely \(\{v_7, v_8, v_{3}, v_{5}\}\), have already been customized. We now customize the cost array \(C(v_1)\) for vertex \(v_1\).
From Figure~\ref{fig:sc-graph}(b), we know \(N^+(v_1) = \{v_7, v_5, v_3\}\).
In Algorithm~\ref{algo:label-construct}, we first consider \( v_7 \in N^+(v_1) \) and iterate over \( \{v_7\} \preceq v_7 \) (Line 6). At Line 7, since \( \omega(v_1, v_7) = 3 \) and \( \delta_{uw} = 0 \) with \( u = w = v_7 \), we compute \( \delta_{vw} = 3 \), where \( v = v_1 \) and \( w = v_7 \). Next, for \( v_5 \in N^+(v_1) \), we iterate over \( \{v_7, v_8, v_3, v_5\} \preceq v_5 \) (Line 6), applying Line 7 to update costs from \( v_1 \) to these vertices. Similarly, for \( v_3 \in N^+(v_1) \) we iterate over \( \{v_7, v_8, v_3\} \preceq v_3 \) (Line 6) and apply Line 7 to update the corresponding costs.
\end{example}

\subsection{Query Processing}

We discuss how to efficiently answer route queries using the customized tree labeling \(L_{\omega}\). This process leverages tree labels to estimate costs and restrict the search to relevant paths, thereby significantly enhancing query efficiency. The hierarchical structure of tree labels is central to this optimization.

Let \(s, t \in V\) be any two vertices and \(CA(s,t)\) their common ancestors as per Definition~\ref{def:td}. Due to the \emph{ancestor separation} property, the optimal cost between \(s\) and \(t\) can be computed as:
\[
d_G^{\omega}(s,t) = \min\{ d_G^{\omega}(s,r) + d_G^{\omega}(t,r) \mid r \in CA(s,t)\}.
\]

Note that our cost arrays do not store distances in $G$.
Instead, it can be shown that $\delta_{vw}=C(v)[\tau(w)]$ is computed as the distance between $v$ and $w$ in the subgraph induced by the vertices in the sub-tree rooted in $w$.
Using the same arguments as in \cite{qiu2022efficient} which takes a similar approach, one can show that
\[
d_G^{\omega}(s,t) = \min\{ \delta_{sr} +  \delta_{tr} \mid r \in \anc(s)\cap\anc(t)\}.
\]

Observe that the cost values used in this are found at the first $h$ positions of $C(s),C(t)$, where $h=|\anc(s)\cap\anc(t)|$. In Algorithm~\ref{algo:query} we compute $h$ as the rank of the lowest common ancestor $l_{st}$ of $s,t$. First we proceed as in \cite{farhan2023hierarchical} and obtain the level of its tree node $f(l_{st})$ as the length of the common prefix of the identifiers $\mathcal{I}(s),\mathcal{I}(t)$, then look up its rank in either of the rank arrays $\Tau(s),\Tau(t)$.

\begin{algorithm}[t]
\caption{Query optimal routes}\label{algo:query}
\SetCommentSty{textit}
\SetKwFunction{FMain}{FindRoute}
\SetKwFunction{FLca}{GetLcaHeight}
\SetKwFunction{FCost}{GetCost}
\SetKwProg{Fn}{Function}{}{end}
\Fn{\FMain{$s,t$}}{
	$h\gets$ \FLca{$s,t$}\\
	$C(s)\gets$ \FCost{$s,h$}\sep
	$C(t)\gets$ \FCost{$t,h$}\\
	\KwSty{return} $\min_{i=1}^{h} C(s)[i] + C(t)[i]$\\
}
\end{algorithm}

\begin{algorithm}[b]
\caption{Cost computation
}\label{algo:label-full}
\SetCommentSty{textit}
\SetKwFunction{FMain}{GetCost}
\SetKwProg{Fn}{Function}{}{end}
\Fn{\FMain{$v,h$}}{
	\If{$C(v)$ is not truncated}{
		\KwSty{return} $C(v)$\\
	}
	\tcp{initialize cost and ancestor arrays of length $\tau(v)$}
	$c\gets [\infty,\dots,\infty]$\sep
	$a\gets [\bot,\ldots,\bot]$\\
	\tcp{follow $S_\omega$ shortcuts upwards while labels are truncated}
	$c[\tau(v)]\gets 0$\sep
	$a[\tau(v)]\gets v$\\
	\For{$\tau(w)=\tau(v)$ down to $1$}{
		$w\gets a[\tau(w)]$\\
		\If{$w=\bot$}{
			\KwSty{continue}\label{ln:cc-continue}\\
		}
		\uIf{$C(w)$ is not truncated} {
			\tcp{update $c$ using $C(w)$}
			\For{$i\in[1,\min(\tau(w) - 1,h)]$\label{ln:cc-loop}}{
				$d\gets c[\tau(w)] + C(w)[i]$\\
				\If{$d\leq c[i]$}{
					$c[i]\gets d$\sep
					$a[i]\gets\bot$\label{ln:cc-loop-end}\\
				}
			}
		}
		\Else{\label{ln:cc-else}
			\tcp{follow $S_\omega$ shortcuts}
			\ForEach{$(n,\omega)\in N^+(w)$}{
				$d\gets c[\tau(w)] + \omega$\\
				\If{$d<c[\tau(n)]$}{
					$c[\tau(n)]\gets d$\sep
					$a[\tau(n)]\gets n$\label{ln:cc-else-end}\\
				}
			}
		}
	}
	\KwSty{return} $c$
}
\end{algorithm}

\begin{figure}
    \centering
    \includegraphics[width=0.6\textwidth]{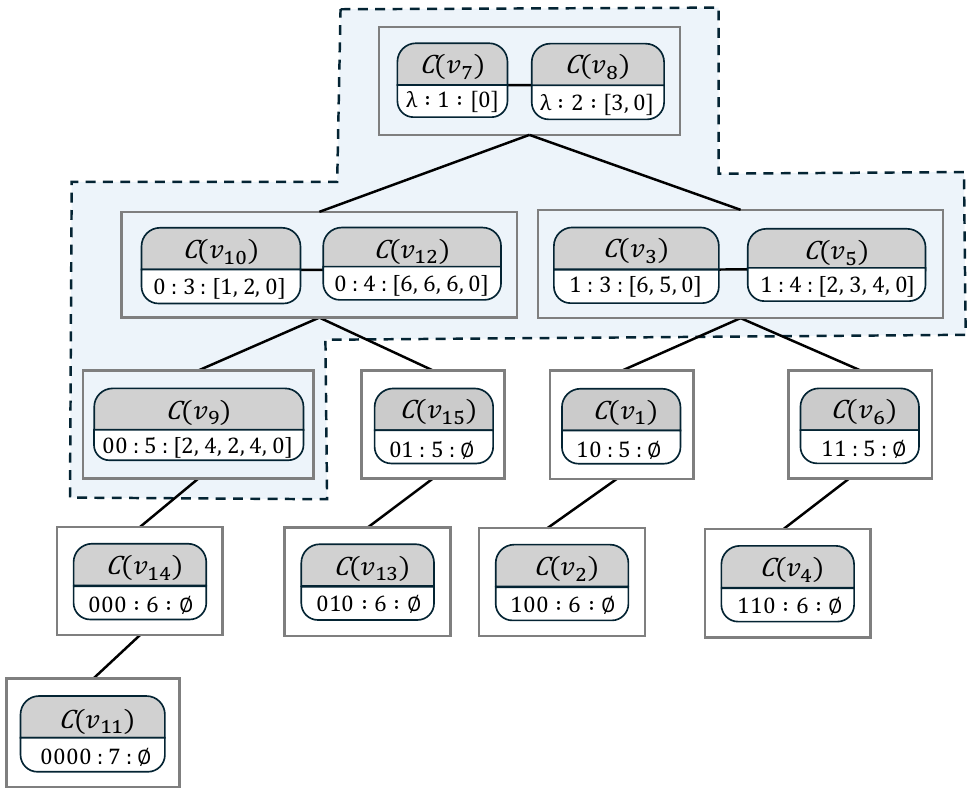}
    \caption{Parameterized labeling \( L^{\theta}_{\omega} \) with \(  \theta=2\), where vertices at the first level are assigned the empty bitstring $\lambda$.}
    \label{fig:parametric}
\end{figure}

\section{Parametric Customization}\label{section:parameterized}
Balancing customization and query performance is crucial for routing systems to support a wide range of applications, from real-time navigation to logistics planning. The tree labeling method introduced in Section~\ref{section:tree-labeling} enables efficient queries, but computing full labels for every vertex can be unnecessarily expensive. We observe that vertices higher in the hierarchy tend to be more central and are involved in queries more frequently, making full labeling beneficial. In contrast, lower-level vertices contribute less and require only minimal labeling. To leverage this distinction, we introduce a parameterized labeling approach that adjusts the level of labeling based on the hierarchical importance of each vertex and enables controllable trade-offs between
customization and query efficiency.

The labels in $L_{\omega}$ serve as pre-computed optimal routes, avoiding the need for upward searches in $S_{\omega}$.
The core idea is to combine \emph{partial} tree labels in the customized tree labeling \( L_{\omega} \) and shortcuts in the customized shortcut graph \( S_{\omega} \) to answer route queries. A higher degree of labeling in \( L_{\omega} \) improves query performance but increases customization time and label size. Conversely, relying more on searching within \( S_{\omega} \) enables faster customization and smaller label sizes but slows query responses. This interplay between \( L_{\omega} \) and \( S_{\omega} \) adds design complexity, requiring careful consideration to ensure accurate routing with minimal computational overhead.

A key challenge here is how exactly to utilize preserved labels when answering queries for lower-level vertices whose labels have been truncated. Although it is possible to simply use our shortcut graph scheme as a fall-back whenever labels are unavailable, such an approach would result in an unacceptable degradation of query performance -- close to that of plain contraction hierarchies, while incurring additional space and customization overheads.

\vspace{0.1cm}
\newcommand{\plabel}{L^{\theta}_{\omega}}
\newcommand{\pcost}{C^{\theta}}
\subsubsection{Parameterized Labeling}
A \emph{parameterized labeling}, denoted as \( \plabel \)  where $\theta\in \mathbb{N}_0$ is a non-negative integer, modifies \( L_{\omega} \) by altering its cost arrays, truncating those near the bottom of the hierarchy. Let $desc(v)=\{w\in V\mid v\otree w\}$ denote the descendants of $v$.
The cost arrays in \( \plabel \) are defined as \( \pcost = \{ \pcost(v) \mid v\in V, |desc(v)| > \theta \} \). Essentially, labels in $\plabel$ are retained unchanged ($\pcost(v)=C(v)$) at the upper levels of the tree hierarchy, but \emph{truncated} ($\pcost(v)=\emptyset$) at the lower levels.

This design of parameterized labeling is further supported by several observations: (1) \emph{Local searches:} Vertices at lower levels tend to have few upward neighbors in $S_{\omega}$, limiting the effort required for local searches in $S_{\omega}$. (2) \emph{Vertex distribution:} Most vertices appear in lower levels, so truncating them reduces label size significantly. (3) \emph{Label sizes:} Cost arrays make up the bulk of label sizes. Vertices at lower levels have more ancestors, leading to larger labels.

\begin{example}
Consider Figure~\ref{fig:labeling}, where vertices in dark blue represent untruncated vertices, while the vertices in light blue are truncated for \( \theta=2\). Consider for example vertex \(v_{12}\): we have \( desc(v_{12})=\{v_9,v_{11},v_{12},v_{13},v_{14},v_{15}\} \) of size $>\theta$,
so its full label information consisting of \(\mathcal{I}(v_{12})=0, \Tau(v_{12})=[2,4]\) and \(C(v_{12})=[6,6,6,0]\) is preserved. In contrast, for vertex \(v_1\) we have \( desc(v_1)=\{v_1,v_2\} \) of size $\leq\theta$,
so its cost array is not stored. We still maintain \(\mathcal{I}(v_1)=10\) and \(\Tau(v_1)=[2,4,5]\).
\end{example}

\subsubsection{Integrated Querying}
Using a parameterized labeling \( \plabel \), we propose an integrated query algorithm that combines shortcut exploration and label-based cost lookup. For cases where costs $C(s)$ and/or $C(t)$ have been truncated due to parameterization, the algorithm computes these costs using both $S_{\omega}$ and $\plabel$. Once computed, $C(s)$ and $C(t)$ can be used to find an optimal route as before. A high-level description is provided in Algorithm~\ref{algo:query}. In Line 3, we pass the height $h$ of the lowest common ancestor to Algorithm~\ref{algo:label-full}, since we only require the first $h$ cost values and will use this fact to limit computational work.

At the core of integrated querying lies cost computation, as described in Algorithm~\ref{algo:label-full}. For a query vertex $v$ whose cost has been truncated (i.e., $\pcost(v)=\emptyset$), we perform a combination of online and offline searches with online search being conducted on the shortcut graph $S_{\omega}$ and offline search on the truncated hub labeling $\plabel$. We start by initializing a cost array $c$ and an ancestor array $a$ (Lines 4--5). We then perform an upward search on $S_{\omega}$ until we reach a vertex whose cost has not been truncated (Lines \ref{ln:cc-else}--\ref{ln:cc-else-end}). From here, we replace the remaining steps of the upward search with more efficient cost lookups in $\plabel$ (Lines 10--14). Observe that in Lines \ref{ln:cc-loop}--\ref{ln:cc-loop-end} we do not update costs to all ancestors, but only for those up to height $h$, as these are the only ones used by Algorithm~\ref{algo:query}. Note further that we may skip distance updates for an ancestor $w=\bot$ in Algorithm~\ref{ln:cc-continue} since in those cases $w$ was not reached through upward search in $S_{\omega}$, and while the ancestor at rank $\tau(w)$ (whose identity is unknown) may well appear on a shortest path $p$ from $v$ to another ancestor, this path must have been considered previously when processing $C(w')$, where $w'$ is the first untruncated ancestor of $v$ in $p$.
The deletion of ancestors in Algorithm~\ref{ln:cc-loop-end}, which is done to prevent needless processing, can be justified in the same way.

If $\theta=0$ or the costs of a query pair $(s,t)$ have not been truncated (i.e., $\pcost(s)\neq\emptyset$ and $\pcost(t)\neq\emptyset$), integrated querying is reduced to querying over the full labeling $L_{\omega}$.
Conversely, for $\theta=\infty$, integrated querying reduces to bidirectional search over $S_{\omega}$ as in CCH.

\begin{example}
Consider a query $(v_{14},v_{13})$, which is processed using the parameterized labeling $\plabel$ shown in Figure~\ref{fig:labeling}.
Since both $C^{\theta}(v_{14})$ and $C^{\theta}(v_{13})$ are truncated, their costs will be computed up to height $h=3$, the rank of their lowest common ancestor $v_{12}$. For $v=v_{14}$ in Figure~\ref{algo:label-full}, the computation begins by following the upwards neighbors $\{v_9,v_{12}\}$ of $v_{14}$ in $S_{\omega}$ as shown in Figure~\ref{fig:sc-graph}(b). The cost and ancestor arrays are then updated to $c=[\infty,\infty,\infty,2,2,0]$ and $a=[\bot,\bot,\bot,v_{12},v_{9},v_{14}]$ in Lines \ref{ln:cc-else}--\ref{ln:cc-else-end}.
In the next iteration ($w=v_9$) we find that $C(v_9)$ is not truncated, and update the distance array to $c=[4,6,4,2,2,0]$ in Lines \ref{ln:cc-loop}--\ref{ln:cc-loop-end}. This is repeated for $w=v_{12}$ but does not cause any changes to $c$. For ancestors $w\in\{v_7,v_8,v_{10}\}$, $a[\tau(w)]$ is unknown ($\bot$), so no further updates are performed. The process of cost computation for \(C(v_{13})\) is similar, with the key difference being that \(C^{\theta}(v_{15})\) is truncated. As a result, upward search in \(S_{\omega}\) is performed during the second iteration. This leads to \( a = [\bot, v_8, v_{10}, v_{12}, v_{15}, v_{13}] \) being almost completely known, with no iterations skipped.
\end{example}

\section{Shortest Path Identification}\label{section:routing}

The integrated querying described in Section~\ref{section:parameterized} identifies the hub rank $i^\star = \tau(x)$ that minimizes the total cost of a route between two vertices. While this suffices to compute distances, it does not reveal the actual path or even the identity of the hub vertex $x$. In the following we extend CTL to return shortest paths rather than distances. For this we introduce additional structures that enable efficient path reconstruction.

\subsection{Basic Valley-Path Unpacking Approach}
Our basic approach for unpacking shortcuts into valley paths follows a similar principle to \cite{dibbelt2016customizable}, but is adapted to handle both truncated and untruncated settings.  
For every shortcut $(u,v)$ we store a triangle node $w$ such that $d(u,w) + d(w,v) = d(u,v)$.  
Such a node $w$ -- if it exists -- is determined during shortcut construction, which explicitly relies on these triangles.  
If no such triangle node exists because $(u,v) \in G$, we store a \texttt{NAN} (not-a-node) identifier.  
Unpacking then proceeds recursively by replacing each shortcut with the corresponding valley path.

\subsection{Path Reconstruction}
Building on the shortcut-level unpacking above, we now consider how to reconstruct a complete path between two query vertices $s$ and $t$.  
An optimal path can be viewed as a chain of valley paths: from $s$ up to its hub node, from $t$ up to its hub node, and finally connecting at the hub itself.  
Since each valley path is represented by a shortcut, full path reconstruction reduces to recursively expanding these shortcuts.
During the initial truncated phase of a query, the sequence of shortcuts is constructed explicitly and tracked via \emph{endpoint predecessor maps}, which directly provide the intermediate vertices.  
In contrast, during the untruncated phase, each label entry may encode an entire sequence of shortcuts without exposing ancestor identifiers.  
To enable unpacking in this case, we introduce \emph{path arrays}, precomputed during customization, which store the necessary endpoint information to expand these shortcut sequences.

\begin{definition}[Path Array]
For each vertex $w$ with an untruncated label, we maintain a \emph{path array}
\[
P(w)[i], \quad 1 \leq i < \tau(w),
\]
where $P(w)[i]$ denotes the first shortcut endpoint on an optimal $\omega$-path from $w$ to its ancestor of rank $i$.
\end{definition}

Path arrays are built alongside the cost arrays $C(w)$ during customization: whenever a cost entry $C(w)[i]$ is updated through a shortcut $(w,u)$, we set $P(w)[i]=u$. This doubles the space requirement of untruncated labels but does not noticeably increase customization time, since the updates follow the same control flow as cost propagation.

While path arrays handle the untruncated part of the query (Algorithm~\ref{algo:label-full}, Lines 10--14), the truncated part (Algorithm~\ref{algo:label-full}, Lines 15--19) explicitly follows shortcuts in $S_\omega$. In this case, endpoints can be recorded directly during query execution.

\begin{definition}[Endpoint Predecessor Map]
For each query endpoint $v\in\{s,t\}$ we maintain a temporary array $\mathrm{EP}_v[\cdot]$ indexed by ranks. Whenever a shortcut $(w,n)$ is relaxed during the truncated phase, we set $\mathrm{EP}_v[\tau(n)] = w$.
\end{definition}

The two structures complement one another: the truncated part of the query produces explicit endpoint information, while the untruncated part can be reconstructed via path arrays. Combining both yields complete endpoint chains.

\begin{lemma}[Endpoint Chains]\label{label:endpoint_chain}
Let $i^\star=\tau(x)$ be the hub rank identified by the integrated query. Then, for each query vertex $v$, one can construct a strictly rank-decreasing sequence $v = e_0, e_1, \ldots, e_k = x$ such that every $(e_j,e_{j+1})\in S_\omega$ is a shortcut lying on an optimal path. Specifically:
\begin{itemize}
    \item If $e_{j}$ was reached during the truncated phase, then $e_{j-1}=\mathrm{EP}_v[\tau(e_{j})]$.
    \item If $e_{j}$ is not truncated, successors are obtained from the path array, i.e., $e_{j+1}=P(e_j)[i^\star]$.
\end{itemize}
\end{lemma}


Note that the first untruncated vertex $e_j=EP[\tau(x)]$ on a shortest path from $v$ to $x$ meets both conditions of Lemma~\ref{label:endpoint_chain}.
We can thus construct the remaining endpoints by applying the Lemma in both directions.
This is sketched in Algorithm~\ref{algo:label-endpoint-chain}.

Once endpoint chains have been computed for both $s$ and $t$, they converge at the hub vertex $x$. Concatenating them yields a path in the shortcut graph $S_\omega$. To obtain the actual route in the original graph $G$, each shortcut is recursively unpacked. This approach is sketched in Algorithm~\ref{algo:path_reconstruction}.

\begin{algorithm}[t]
\caption{Construct endpoint chain}
\label{algo:label-endpoint-chain}
\SetCommentSty{textit}
\SetKwFunction{FMain}{EndpointChain}
\SetKwProg{Fn}{Function}{}{end}
\Fn{\FMain{$v,i^\star$}}{

    $\mathcal{E}_v \gets [\,]$\;
    $e \gets \mathrm{EP}_v[i^\star]$\;
    \While{$e\neq v$}{
        $e\gets\mathrm{EP}_v[\tau(e)]$\;
        append $e$ to $\mathcal{E}_v$\;
    }
    reverse($\mathcal{E}_v$)\;
    $e \gets \mathrm{EP}_v[i^\star]$\;
    append $e$ to $\mathcal{E}_v$\;
    \While{$\tau(e) \neq i^\star$}{
        $e \gets P(e)[i^\star]$\;
        append $e$ to $\mathcal{E}_v$\;
    }
    \Return{$\mathcal{E}_v$}\;
}
\end{algorithm}

\begin{algorithm}[t]
\caption{Path reconstruction for CTL}
\label{algo:path_reconstruction}
\SetCommentSty{textit}
\SetKwFunction{FMain}{GetPath}
\SetKwProg{Fn}{Function}{}{end}
\Fn{\FMain{$s, t$}}{
$(d^\star,i^\star) \gets$ \textsc{IntegratedQuery}$(s,t)$\\
$\mathcal{E}_s \gets$ \textsc{EndpointChain}$(s,i^\star)$\\
$\mathcal{E}_t \gets$ \textsc{EndpointChain}$(t,i^\star)$\\
$\mathcal{E}_{st} \gets$ concatenate $\mathcal{E}_s$ with reversed $\mathcal{E}_t$ at hub $x$\\
$p_{st} \gets [s]$ \tcp{init $s$--$t$ path}
\ForEach{shortcut $(u,v)$ in $\mathcal{E}_{st}$}{
    $\mathcal{P}_{uv} \gets$ recursively unpack $(u,v)$\\
    Append $\mathcal{P}_{uv}[2:\,]$ to $p_{st}$ \tcp{skip first node to avoid duplication}
}
\Return{$(p_{st},d^\star)$}\;
}
\end{algorithm}

\vspace{0.15cm}
\noindent\textbf{Complexity Analysis.} Path queries preserve the efficiency of distance queries. The asymptotic cost of computing distances remains unchanged. Endpoint reconstruction requires only $O(h)$ additional steps per side, where $h$ is the LCA height, and unpacking is output-sensitive in the path length. Storage increases by a factor of two for untruncated labels due to path arrays, while customization time is essentially unaffected.

\subsection{Extended Valley-Path Unpacking Approach}
\label{subsec:extended-valley-unpacking}
In the basic path unpacking approach, once an optimal hub rank is identified, the path query reconstructs a sequence of shortcuts on the optimal route and then recursively unpacks each shortcut $(u,v)\in S_\omega$ using its downward triangle node $z$, replacing $(u,v)$ by $(u,z)$ and $(z,v)$ until only base edges remain. While correct, this recursive unpacking frequently dominates query time: valley paths can be long, and traversing them recursively incurs repeated lookups and pointer chasing.

A natural alternative is to store the full valley path for each shortcut explicitly. However, this has two drawbacks: (i) it significantly increases memory usage, as most shortcuts would need to store several intermediate nodes, and (ii) the number of stored nodes varies, complicating memory layout and access. 

We therefore introduce a hybrid scheme that balances space and time efficiency. Short valley paths are stored explicitly in-line, while long valley paths are represented in unpacked form, though in addition to the triangle node we also store pointers to the data needed to unpack them. This retains fixed storage per shortcut, provides fast access for many paths, and reduces recursion depth for the remaining ones.

\medskip
\noindent\textbf{Data structure.}
Our storage structure is illustrated in Figure~\ref{fig:path-data}. Each shortcut record occupies a fixed block of six 32-bit words (24 bytes). These words are interpreted in one of two ways:

\begin{itemize}
    \item \emph{Triangle mode.} The first word is set to \texttt{NAN}, a reserved value not corresponding to any valid node ID. The next word stores the downward triangle node $z$, and the remaining four words are used as two 64-bit pointers to the records of the child shortcuts $(u,z)$ and $(z,v)$. This corresponds to the classical recursive unpacking approach, but with direct pointers to child records to reduce dictionary lookups.
    \item \emph{Inline-path mode.} If the first word is not \texttt{NAN}, then the six words are interpreted as an array of up to six intermediate nodes of the valley path. Endpoints $u,v$ are implicit. Paths shorter than length six are padded with \texttt{NAN} values. To disambiguate an inline empty path (whose first word would also be \texttt{NAN}), the second word is checked: if it is also \texttt{NAN}, the record encodes an empty list.
\end{itemize}

This fixed-layout approach is realized in our implementation with a \texttt{union} structure in C++. Examination of the first two words suffices to decide which interpretation applies. In practice, the union avoids space overhead, since the triangle mode already requires alignment padding for two 64-bit pointers. Hence storage costs do not increase over the baseline.

\begin{example}
Consider the following path records:
\begin{center}
\begin{enumerate}
\item \texttt{NAN | 006 | 278 | 334 | 456 | 598}
\item \texttt{006 | 005 | NAN | NAN | NAN | NAN}
\item \texttt{NAN | NAN | NAN | NAN | NAN | NAN}
\end{enumerate}
\end{center}
The first case encodes unpacking information, consisting of triangle node $v_6$ and two 64-bit pointers 278334 and 456598.
Cases two and three show intermediate node lists of length 2 and 0 for the shortcuts $v_3-v_7$ and $v_5-v_7$ in Figure~\ref{fig:sc-graph}, respectively.
\end{example}

\begin{figure}[t]
\centering
\newcolumntype{C}{>{\centering\let\newline\\\arraybackslash\hspace{0pt}}m{9mm}}
\begin{tabular}{|C|C|C|C|C|C|}
\hline
\texttt{NAN} & node & \multicolumn{2}{c|}{path-record ptr} & \multicolumn{2}{c|}{path-record ptr} \\
\hline
node? & node? & node? & node? & node? & node? \\
\hline
\end{tabular}
\caption{Fixed-size storage for shortcuts. Each record uses six 32-bit words, interpreted either in triangle mode (top row) or inline-path mode (bottom row).}
\label{fig:path-data}
\end{figure}

\noindent\textbf{Construction.}
During customization, when each shortcut $(u,v)$ is formed from its downward triangle node $z$, we attempt to recover the underlying valley path. If the path length is at most six, it is stored inline; otherwise, the record is set to triangle mode with direct pointers to its two child shortcuts. This process requires no additional asymptotic time beyond what CTL already spends to establish triangle relationships.

\vspace{0.15cm}
\noindent\textbf{Unpacking.}
Path reconstruction proceeds as follows. For each shortcut on the endpoint chain discovered by the integrated query, its record is inspected:
\begin{itemize}
    \item If the record is in inline mode, the stored sequence of nodes is returned directly as the valley path.
    \item If the record is in triangle mode, the stored child pointers are used to recursively expand $(u,z)$ and $(z,v)$, and their expansions are concatenated.
\end{itemize}
This approach is sketched as Algorithm~\ref{alg:expand-hybrid}.
Compared to the basic approach, unpacking is faster in two ways: first, many shortcuts terminate immediately in inline mode; and second, triangle-mode records follow pointers directly without dictionary lookups.

\begin{algorithm}[t]
\caption{Hybrid shortcut unpacking}
\label{alg:expand-hybrid}
\SetCommentSty{textit}
\SetKwFunction{FMain}{ExpandShortcut}
\SetKwProg{Fn}{Function}{}{end}
\Fn{\FMain{$u,v$}}{
    $r \gets \mathsf{Rec}(u,v)$ \tcp{path record}
    \uIf{$r[0] \neq \texttt{NAN}$}{
        \Return $[u, r[1], r[2], \ldots, r[k], v]$
    }
    \uElseIf{$r[1] = \texttt{NAN}$}{
        \Return $[u,v]$
    }
    \Else{
        $z \gets r[1]$\\
        $p_1 \gets$ \textsc{ExpandShortcut}$(u,z)$ via $r$.\texttt{payload$_1$}\\
        $p_2 \gets$ \textsc{ExpandShortcut}$(z,v)$ via $r$.\texttt{payload$_2$}\\
        \Return $p_1 \,\Vert\, p_2[2..]$ \tcp{skip first node of $p_2$}
    }
}
\end{algorithm}

\subsection{Path Array Variants} 
In the basic valley-path unpacking approach, path arrays store only the endpoint of the first shortcut. This keeps the index compact but requires extra work during path reconstruction. A natural extension is to store the full shortcut record, as described in Section~\ref{subsec:extended-valley-unpacking}. We call this approach \emph{Extended Path Arrays}. Although this increases storage requirements, it makes path reconstruction faster since the shortcut records can be accessed directly when unpacking endpoint chains. 

Another alternative that aims in the opposite direction is to not use any path arrays at all, reducing index size at the cost of higher query times. In this case, the $\mathrm{EP}_v$ lookups in Algorithm~\ref{algo:label-endpoint-chain} are replaced by iterating over the upward neighbors of $v$. The correct shortcut is then identified by checking whether it lies on a shortest path using the distance information stored in the labels. This procedure, described in Algorithm~\ref{algo:get-ep}, turns the first while loop of endpoint chain construction into an $A^*$ search with perfect distance bounds.  

\begin{algorithm}[t]
\caption{Shortcut Endpoint Lookup}
\label{algo:get-ep}
\SetCommentSty{textit}
\SetKwFunction{FMain}{GetEP}
\SetKwProg{Fn}{Function}{}{end}
\Fn{\FMain{$v,i$}}{
    \ForEach{$(u,c)\in up(v)$}{
        \If{$c + C(u)[i] = C(v)[i]$}{
            \Return $u$
        }
    }
}
\end{algorithm}

We thus obtain a total of six potential variations of our approach, which differ by what path information is stored in our shortcut graph and path arrays.
For shortcuts we have the following storage options:
\begin{itemize}
\item[(b)] \textbf{B}asic path information: the triangle node only.
\item[(e)] \textbf{E}xtended path information, as shown in Figure~\ref{fig:path-data}.
\end{itemize}
For path arrays we consider three options:
\begin{itemize}
\item[(n)] \textbf{N}o path arrays used.
\item[(b)] \textbf{B}asic path information: the endpoint only.
\item[(e)] \textbf{E}xtended path information, as well as endpoint.
\end{itemize}

Each of the resulting six combinations offers different trade-offs between storage space and query speed.
However we immediately discard the combination (b,e), indicating basic information for shortcuts but extended information for path arrays, since it would result in the worst of both worlds: high storage costs due to extended path arrays, and slow query times as most of the unpacking will still be done using basic information only.
The remaining five options are shown in Table~\ref{fig:variants}(a).

\subsection{Analysis}\label{S:analysis}

Path unpacking is the main bottleneck for answering path queries, especially for long-range queries.
In the following we present a brief complexity analysis based on the different variants of path information stored.

Let $\mathfrak{p}$ denote the length of the shortest path $p$ returned as query answer, $\mathfrak{s}$ the length of its corresponding valley-chain path and $\mathfrak{u}$ the maximum number of upward neighbors in the shortcut graph $S$.
Note that $\mathfrak{s}$ and $\mathfrak{u}$ are both bounded by (twice) the height of the tree decomposition (maximum number of ancestors), and $\mathfrak{s}\leq\mathfrak{p}$.

During path unpacking we first recover the valley-chain path using information stored in $L$, then unpack valley paths using data from $S$.
While the number of steps for these phases is always $\mathfrak{s}$ and $\mathfrak{p}$, respectively, the computation cost of each unpacking step depends on the path information stored in $L$ and $S$.
Without any path information in $L$, identifying the next shortcut on the valley-chain path, starting at $v$, requires a sequential search through $up(v)$ in time $O(\mathfrak{u})$.
If both endpoints of a shortcut are known, finding the corresponding shortcut containing further path information in $S$ can be achieved in $O(\log\mathfrak{u})$, provided $up(v)$ is sorted.
Pointers to the next path record enable constant-time lookup.
Put together, this results in the following complexities for path unpacking:
\[
\setlength{\arraycolsep}{1em}
\begin{array}{lll}
\text{CTL}_{bn}: O(\mathfrak{p}\cdot\log\mathfrak{u} + \mathfrak{s}\cdot\mathfrak{u}) &
\text{CTL}_{bb}: O(\mathfrak{p}\cdot\log\mathfrak{u}) \\
\text{CTL}_{en}: O(\mathfrak{p} + \mathfrak{s}\cdot\mathfrak{u}) &
\text{CTL}_{eb}: O(\mathfrak{p} + \mathfrak{s}\cdot\log\mathfrak{u}) &
\text{CTL}_{ee}: O(\mathfrak{p})
\end{array}
\]

While the complexity for looking up shortcut path information based on two endpoints could be reduced to $O(1)$ using hash tables, in practice the overheads this incurs are not worth it, as sets of upward neighbors are fairly small (cf. Figure~\ref{fig:shortcut_dist}).
In our implementation we do not even perform binary search for CTL$_{eb}$, using basic endpoint information only to avoid additional (constant-time) label lookups.

Storing short valley paths inline does not improve the asymptotic complexity of path reconstruction, but recursion depth is significantly reduced in practice.
Customization overheads for path records do not increase complexity as record construction only requires copying small amounts of readily available data.
Storage costs for extended records are higher than for basic ones though, as each extended path record requires $6 \times 32$ bit, compared to just 32 bit for storing triangle nodes only.

\section{Batch Processing Approach}\label{sec:batch_processing}
Road networks contain motorways and other roads that feature in many shortest paths, and thus will be unpacked frequently.
Thus we can reduce the cost for path unpacking by processing queries in batches and leveraging overlaps between shortest paths returned.
This is an optimization to query answering which requires no changes to the underlying index.

We first compute shortcut chains from between hub nodes and query endpoints as before (two per query).
Whenever unpacking a new shortcut chain, we identify a shortcut chain with maximal common prefix amongst those already unpacked.
This allows us to simply copy the unpacked common prefix, reducing the workload to unpacking the remaining suffix.

To efficiently identify the most similar shortcut chain previously unpacked, we process them in lexicographical order (from hub to query endpoint).
This ensures that the most similar shortcut chain is always the last one processed.

\begin{example}
Consider the graph from Figures \ref{fig:example_road_network} and \ref{fig:sc-graph} with three queries $(v_1,v_{13}),(v_1,v_{14}),(v_2,v_{15})$.
These generate the following shortcut chains, respectively:
\begin{align*}
(v_1,v_{13}):{} & [v_8,v_5,v_1], [v_8,v_{13}]\\
(v_1,v_{14}):{} & [v_7,v_1], [v_7,v_9,v_{14}]\\
(v_2,v_{15}):{} & [v_8,v_5,v_1,v_2], [v_8,v_{15}]
\end{align*}
Sorting shortcut chains lexicographically gives us
{\small\[
[v_7,v_1],
[v_7,v_9,v_{14}],
[v_8,v_5,v_1],
[v_8,v_5,v_1,v_2],
[v_8,v_{13}],
[v_8,v_{15}]
\]}%
leaving $[v_8,v_5,v_1]$ and $[v_8,v_5,v_1,v_2]$ adjacent.
\end{example}

Even stronger, we can show that this approach to reusing common prefixes is optimal in a sense.
Let $S$ be a list of sequences, and let $\rho(S[i],S[j])$ denote the length of the maximal common prefix of the sequences at indices $i,j$.
Define the \emph{overlap} of $S[j]$ w.r.t. $S$ as
\[
\mathcal{O}_S(j) := \max_{i<j} \rho(S[i],S[j])
\]
and the \emph{overlap} of $S$ as
\[
\mathcal{O}(S) := \sum_{j>0} \mathcal{O}_S(j)
\]
which measures the amount of unpacking work we can save, based on a particular unpacking order.

\begin{theorem}\label{th:overlap}
Let $L$ be the lexicographically ordered permutation of $S$. Then
\[
\mathcal{O}(S) = \mathcal{O}(L) = \sum_{j>0} \rho(L[j-1],L[j])
\]
\end{theorem}

\begin{proof}
Consider a permutation of $S$ which only swaps two adjacent elements $S[j]$ and $S[j+1]$, so only the overlaps of $j,j+1$ might change.
Let $P$ be the maximal common prefix of $S[j],S[j+1]$.
If $P$ is a prefix for some $S[i]$ with $i<j$ then the overlaps of $S[j], S[j+1]$ do not change.
Otherwise let $P'$ the maximal prefix of $P$ which is also a prefix of some $S[i]$ with $i<j$.
Then the overlaps of $S[j],S[j+1]$ are $|P'|,|P|$ in $S$ and swap to $|P|,|P'|$ in the permutation, so the total overlap does not change.
As any permutation can be obtained by a series of swaps of adjacent elements, this shows the first equality.
The second equality follows from $L$ being lexicographically ordered.
\end{proof}

Theorem~\ref{th:overlap} shows that (1) the total overlap of a list of sequences is order-independent, and (2) the lexicographical ordering makes identification of the prior sequence with maximal common prefix a trivial task.

\begin{algorithm}[t]
\caption{Query batch processing}\label{algo:batch}
\SetCommentSty{textit}
\SetKwFunction{FMain}{GetRoutesInBatch}
\SetKwProg{Fn}{Function}{}{end}
\Fn{\FMain{$B$}}{
	$S\gets []$ \tcp{empty list of shortcut chains}
	\ForEach{$(s,t)\in B$}{
		$h\gets$ hub on shortest $s-t$ path\\
		find shortcut chains $c_s=[h,\dots,s]$ and $c_t=[h,\dots,t]$\\
		add $c_s$ and $c_t$ to $S$\\
	}
	sort $S$ lexicographically\\
	$prev\gets[]$, $prev\_unpacked\gets[]$\\
	\ForEach{$sc\in S$ in order}{
		$P\gets$ maximal common prefix of $sc,prev$\\
		\If{$|P|>1$}{
			copy $prev\_unpacked$ up to $P[last]$\\
		}
		unpack remaining shortcuts in $sc$\\
		$prev\gets sc$, $prev\_unpacked\gets$ unpacked($sc$)\\
	}
	concatenate unpacked paths at hubs\\
	\KwSty{return} concatenated paths\\
}
\end{algorithm}

The resulting approach to batch processing is outlined in Algorithm~\ref{algo:batch}.
Rather than actually sorting shortcut chains, we sort indexes referencing them, which also enables us to easily identify unpacked paths to concatenate at the end by their index.

Note that instead of orienting shortcut chains from hub to query endpoint, we could also orient them from endpoint to hub (or equivalently use reverse lexicographical ordering and common suffixes).
However, shortcut chains are far more likely to agree on their hubs, and shortcuts close to the hubs tend to represent longer valley paths, so a larger amount of unpacking work is avoided by copying common prefixes.

\section{Key Variants}\label{section:variants}

\subsubsection{Parallel Customization}
The customization of tree labeling can be parallelized as follows. Vertices \(v \in V\) are divided into groups \(G_1 \dots G_k\) based on their rank \(i = \tau(v)\). These groups are processed in descending order of \(i\) for customizing the shortcut graph and in ascending order of \(i\) for customizing the labeling, enabling parallel execution. Synchronization between groups is achieved using a barrier to align threads. Each thread writes only to the shortcut and label it is currently processing and reads only from shortcuts and labels in strictly lower (or upper) groups, avoiding read/write conflicts and eliminating the need for locks or atomic operations. This approach parallelizes Line 4 of Algorithms~\ref{algo:custom-ch} and \ref{algo:label-construct}, processing vertices in increasing order of levels \(i\).

\subsubsection{Directed Road Networks}\label{subsec:directed_road_networks}
Our method can be easily adapted for directed road networks by modifying the tree labeling scheme \(L\), specifically the cost arrays \(C\). For each vertex \(v \in V\), we create \emph{forward} and \emph{reverse} cost arrays to store costs for both directions during customization with respect to a given metric. This may be carried out by running Algorithm~\ref{algo:label-construct} for both forward and reverse directions. Then, we can use \om$^-$ and \om$^+$ twice to maintain $L$, once on the forward labels and again on the backward search. For directed versions of dynamic road networks, existing labelling-based methods require increased memory to store precomputed labels. However, road networks are often nearly undirected, with a few notable exceptions (e.g., Stockholm). In such cases,  two distances stored within each label are often identical.

\section{Experiments}\label{section:experiments}

\subsection{Setup}
We conducted experiments on a Linux server equipped with a 13th Gen Intel(R) Core(TM) i9-13900K processor (32 CPUs) and 128GB of main memory. We implemented the proposed method in C++20 and compiled with the GNU C++ compiler 11.4.0 using \textit{-O3} and \textit{-march=native} optimizations. 

\subsubsection{Datasets}
We use thirteen real-world road networks summarized in Table~\ref{table:datasets}. Twelve are from the US, sourced from the 9th DIMACS Implementation Challenge~\cite{demetrescu2009shortest}, and one is from Western Europe, managed by PTV AG~\cite{ptvplanung}.
For some of the larger datasets approaches can run out of memory, which we indicate using a -- symbol.

\begin{table*}[t]
\centering
\caption{Summary of 13 real-world road networks.}
\label{table:datasets}
\small
\begin{tabular}{| l l | r r r r |} 
    \hline
    Network (Net.) & Region & $|V|$ & $|E|$ & diam. & Memory \\
    \hline
    NY & New York City & 264,346 & 733,846 & 720 & 17 MB \\
    BAY & San Francisco & 321,270  & 800,172 & 721 & 18 MB \\
    COL & Colorado & 435,666 & 1,057,066 & 1,245 & 24 MB \\
    FLA & Florida & 1,070,376 & 2,712,798 & 2,058 & 62 MB \\
    NW & Northwest US & 1,207,945 & 2,840,208 & 1,994 & 62 MB \\
    NE & Northeast US & 1,524,453 & 3,897,636 & 2,098 & 62 MB \\
    CAL & California & 1,890,815 & 4,657,742 & 2,315 & 107 MB \\
    LKS & Great Lakes & 2,758,119 & 6,885,658 & 4,127 & 107 MB \\\hline
    EUS & Eastern USA & 3,598,623 & 8,778,114 & 4,461 & 201 MB \\
    WUS & Western USA & 6,262,104 & 15,248,146 & 4,420 & 349 MB \\
    CUS & Central USA & 14,081,816 & 34,292,496 & 5,533 & 785 MB \\
    USA & United States & 23,947,347 & 58,333,344 & 8,440 & 1.30 GB \\ 
    EUR & Western Europe & 18,010,173 & 42,560,279 & 3,175 & 974 MB \\
    \hline
\end{tabular}
\end{table*}

\subsubsection{Baseline Algorithms}
We compare our method, \om, against \emph{Customizable Contraction Hierarchies (CCH)}~\cite{dibbelt2016customizable}, the state-of-the-art fully customizable technique.
CCH is conceptually identical to \om\ with $\theta=\infty$, but the two approaches use different orderings and vary in terms of optimization details as well, such as triangle precompu\-ta\-tion for CCH which speeds up customization but increases index size. To ensure a fair comparison independent of these details, our implementation is used for both.

We also compare against \emph{Hierarchical 2-Hop (H2H)} with shortest paths~\cite{ouyang2023hierarchy}, the state-of-the-art technique for path computation without customization.
While not a direct competitor due to the lack of efficient customization, it provides an additional baseline for comparison of query speed, index size and construction time.

For our method, we use the notations shown in Figure~\ref{fig:variants} to denote the different variants (e.g. CTL$_{eb}$), and use the superscript $^b$ (e.g. CTL$_{eb}^b$) to denote their batch processing counterparts.
Similarly, we denote the sequential and batch variants of Customizable Contraction Hierarchies as CCH and CCH$^b$, respectively.  

Note that Customizable Tree Labeling with the path extension we sketched in \cite{farhan2025customization} corresponds to CTL$_{bb}$.
Since CTL without path information includes Dual-Hierarchy Labeling \cite{farhan2025dual} as the special case $\theta=0$, we are comparing against extensions of DHL augmented with path information as well.
Since DHL is the current state-of-the-art for dynamic road networks, we compare customization against DHL's update times.
While DHL does not maintain path information, the overhead for doing so is negligible.

We exclude CRP~\cite{delling2017customizable} from our comparisons, as CCH has been shown to achieve superior query performance \cite{dibbelt2016customizable,blasius2025customizable}.
We also omit direct comparison against non-customizable approaches other than H2H, including goal-directed search techniques~\cite{goldberg2005computing,maue2010goal,hart1968formal} and hierarchical methods~\cite{sanders2005highway,jung2002efficient}, since these have been substantially outperformed in query time by more recent non-customizable techniques such as Contraction Hierarchies (CH) \cite{geisberger2008contraction,geisberger2012exact} and labeling-based methods \cite{abraham2011hub,luxen2011hierarchy}. 

\subsubsection{Benchmark Generation}
To evaluate customization performance against CCH, we generated five random sets of edge weights based on travel times and report the average customization time.

To evaluate query performance, we sampled 1,000,000 random vertex pairs in each road network. Following~\cite{Pohl1969BidirectionalAH}, we also sampled sets of pairs containing short, medium and long range query pairs. Specifically, for each road network, we generate 10 sets of query pairs $Q_1, \dots, Q_{10}$ as follows: we set $l_{min}$ to be 1000 meters, and set $l_{max}$ to be the maximum distance of any pair of vertices in the network. Let $x = (\frac{l_{max}}{l_{min}})^{1/10}$. For each $1 \leq i \leq 10$, we sample 10,000 query pairs to form each set $Q_i$, in which the distance of the source and target vertices for each query falls in the range $(l_{min} \cdot x^{i-1},\; l_{min} \cdot x^{i}]$. Note that we shall refer to $Q1$--$Q4$, $Q5$--$Q7$ and $Q8$--$Q10$ as short, medium and long range query sets, respectively.

Finally, we compare preprocessing time and labeling size of state-of-the-art methods against those of \om.

\newcommand*{\hdr}[1]{#1}
\newcommand*{\thdr}[1]{\hdr{$\theta{=}#1$}}
\newcommand*{\oom}[1]{\multicolumn{1}{#1}{--}}
\newcommand{\shrinkbox}[3]{\resizebox{\ifdim\width>#1#1\else\width\fi}{#2}{\small#3}}

\newcommand{\HdrTheta}{\thdr{0}&\thdr{2}&\thdr{5}&\thdr{10}&\thdr{20}&\thdr{50}&\thdr{100}}

\newcommand{\TblNet}{\begin{tabular}{|l|}
\hline
\multirow{2}{*}{\hdr{\textbf{Net.}}}\\\\
\hline
NY\\BAY\\COL\\FLA\\NW\\NE\\CAL\\LKS\\
\hline
EUS\\WUS\\CUS\\USA\\EUR\\
\hline
\end{tabular}}

\setlength{\tabcolsep}{4pt}


\newcommand{\TblCustomization}{\begin{tabular}{| c H H c c c c | c |}
\hline
\multicolumn{8}{|c|}{\textbf{Customization Time (CT) [s]}}\\
\hline
\HdrTheta & \textbf{CCH}\\
\hline
0.200&0.186&0.146&0.134&0.123&0.110&0.102&0.081\\
0.149&0.136&0.101&0.091&0.084&0.074&0.070&0.061\\
0.200&0.181&0.129&0.116&0.105&0.093&0.087&0.071\\
0.655&0.573&0.404&0.360&0.326&0.294&0.276&0.239\\
0.672&0.584&0.390&0.349&0.317&0.284&0.266&0.222\\
1.565&1.362&0.935&0.826&0.752&0.660&0.615&0.475\\
1.607&1.398&0.939&0.831&0.751&0.673&0.630&0.511\\
3.783&3.262&2.258&1.962&1.767&1.557&1.436&1.012\\
\hline
4.211&3.601&2.405&2.088&1.831&1.593&1.471&1.139\\
8.352&7.219&4.985&4.346&3.822&3.285&2.994&2.297\\
31.77&27.62&19.75&17.38&15.38&13.22&11.97&7.784\\
49.46&43.75&31.51&27.82&24.65&21.20&19.14&12.37\\
42.69&37.57&27.65&24.46&21.71&18.57&16.59&9.417\\
\hline
\end{tabular}}

\newcommand{\TblQueryDistance}{\begin{tabular}{| c H H c c c c | c | c |}
\hline
\multicolumn{9}{|c|}{\textbf{Query Time (QT) [$\mu$s]}}\\
\hline
\HdrTheta & \textbf{CCH} & \textbf{H2H}\\
\hline
0.097 & 0.439 & 0.989 & 1.283 & 1.607 & 2.115 & 2.559 & 13.97 & 0.275\\
0.102  & 0.426 & 0.934 & 1.152 & 1.389 & 1.735 & 2.013 & 8.310 & 0.276\\
0.125  & 0.512 & 1.147 & 1.433 & 1.721 & 2.089 & 2.382 & 11.80 & 0.482\\
0.184  & 0.652 & 1.424 & 1.812 & 2.234 & 2.801 & 3.288 & 11.79 & 0.483\\
0.149  & 0.632 & 1.436 & 1.813 & 2.228 & 2.783 & 3.245 & 13.25 & 0.576\\
0.222  & 0.774 & 1.733 & 2.223 & 2.784 & 3.617 & 4.451 & 30.71 & 0.658\\
0.259  & 0.852 & 1.831 & 2.313 & 2.872 & 3.623 & 4.350 & 21.99 & 0.709\\
0.315 & 0.991 & 2.013 & 2.606 & 3.197 & 4.163 & 5.183 & 63.51 & 0.809\\
\hline
0.313 & 1.073 & 2.319 & 2.748 & 3.354 & 4.232 & 5.158 & 44.45 & 0.871\\
0.345 & 1.160 & 2.473 & 3.093 & 3.839 & 4.354 & 5.300 & 43.01 & 1.006\\
0.474 & 1.549 & 3.197 & 4.059 & 5.115 & 6.627 & 7.678 & 139.9 & 1.696\\
0.489 & 1.568 & 3.230 & 4.108 & 5.080 & 5.958 & 7.344 & 173.4 & --\\
0.587 & 1.799 & 3.452 & 4.672 & 6.015 & 7.736 & 9.200 & 247.7 & --\\
\hline
\end{tabular}}

\newcommand{\TblQueryBB}{\begin{tabular}{| c H H c c c c |}
\hline
\multicolumn{7}{|c|}{\textbf{CTL$_{bb}$ [$\mu$s]}}\\
\hline
\HdrTheta\\
\hline
14.13&14.77&15.58&15.97&16.31&16.67&17.35\\
15.61&16.22&16.79&17.31&17.27&17.66&18.08\\
25.45&25.92&26.52&27.00&27.21&27.31&27.92\\
34.28&34.45&35.75&35.94&36.62&37.09&37.85\\
45.93&45.76&47.11&47.50&47.78&48.39&48.94\\
46.56&47.33&48.90&49.51&50.44&51.77&53.19\\
58.34&58.36&59.82&60.33&61.24&62.11&63.00\\
97.81&97.41&99.82&100.5&101.4&102.8&104.0\\
\hline
98.93&98.49&99.88&101.4&101.9&103.1&103.9\\
172.9&173.3&176.1&176.1&176.9&178.3&179.2\\
290.7&291.4&296.1&297.0&298.2&301.0&302.5\\
--   &--   &438.4&439.7&439.5&435.1&439.0\\
--   &--   &125.7&127.5&128.6&131.0&133.1\\
\hline
\end{tabular}}

\newcommand{\TblBatchBB}{\begin{tabular}{| c H H c c c c |}
\hline
\multicolumn{7}{|c|}{\textbf{CTL$_{bb}^b$ [$\mu$s]}}\\
\hline
\HdrTheta\\
\hline
3.857&3.604&4.078&4.381&4.673&4.979&5.337\\
4.323&3.824&4.263&4.484&4.713&5.005&5.330\\
5.798&4.842&5.320&5.507&5.802&6.120&6.391\\
6.559&5.647&6.228&6.548&6.841&7.351&7.769\\
8.090&6.549&7.079&7.376&7.697&8.150&8.565\\
7.702&6.903&7.657&8.048&8.566&9.317&9.952\\
9.070&7.744&8.401&8.746&9.173&9.883&10.56\\
11.36&9.603&10.34&10.78&11.31&12.29&13.30\\
\hline
11.67&11.06&11.24&11.65&12.15&12.86&13.79\\
16.35&14.15&14.82&15.37&15.56&16.38&17.24\\
20.44&20.91&22.33&23.03&23.93&25.32&26.70\\
--   &--   &32.15&28.24&28.85&29.93&31.15\\
--   &--   &17.31&18.54&19.14&20.66&22.40\\
\hline
\end{tabular}}

\newcommand{\TblQueryCCH}{\begin{tabular}{|c|}
\hline
\multirow{2}{*}{\textbf{CCH}}\\\\
\hline
35.71\\29.04\\40.05\\47.76\\69.90\\96.31\\94.92\\205.9\\
\hline
175.8\\251.5\\556.7\\757.2\\524.1\\
\hline
\end{tabular}}

\newcommand{\TblBatchCCH}{\begin{tabular}{|c|}
\hline
\multirow{2}{*}{\textbf{CCH$^b$}}\\\\
\hline
18.40\\14.44\\20.43\\17.78\\20.76\\44.79\\35.57\\78.86\\
\hline
61.28\\63.36\\199.4\\292.6\\304.2\\
\hline
\end{tabular}}

\newcommand{\TblQueryHH}{\begin{tabular}{|c|}
\hline
\multirow{2}{*}{\textbf{H2H}}\\\\
\hline
21.96\\26.00\\47.25\\71.09\\101.4\\96.90\\139.8\\223.8\\
\hline
225.6\\435.7\\754.1\\--\\--\\
\hline
\end{tabular}}

\newcommand{\TblQueryEE}{\begin{tabular}{| c H H c c c c |}
\hline
\multicolumn{7}{|c|}{\textbf{CTL$_{ee}$ [$\mu$s]}}\\
\hline
\HdrTheta\\
\hline
2.702&3.637&4.755&5.304&5.778&6.229&6.475\\
3.100&3.860&4.925&5.339&5.584&5.717&5.879\\
4.785&5.846&7.137&7.582&7.889&7.918&8.131\\
6.097&7.061&8.446&9.031&9.500&10.03&10.29\\
8.062&9.040&10.50&11.01&11.51&11.81&11.67\\
7.548&9.425&11.77&12.64&13.60&14.68&15.64\\
9.771&11.47&13.53&14.45&15.15&16.18&17.09\\
13.90&16.89&20.51&21.93&23.08&25.02&26.53\\
\hline
14.00&16.95&20.34&21.58&22.82&24.35&25.86\\
24.08&26.97&30.70&32.13&33.30&34.99&36.61\\
--   &--   &47.78&50.52&52.83&55.3&57.18\\
--   &--   &--   &82.96&85.09&87.99&82.24\\
--   &--   &--   &32.67&29.62&30.90&32.16\\
\hline
\end{tabular}}

\newcommand{\TblBatchEE}{\begin{tabular}{| c H H c c c c |}
\hline
\multicolumn{7}{|c|}{\textbf{CTL$_{ee}^b$ [$\mu$s]}}\\
\hline
\HdrTheta\\
\hline
3.907&4.001&4.691&5.122&5.611&6.130&6.567\\
4.416&4.435&5.110&5.419&5.710&6.033&6.311\\
5.587&5.411&6.114&6.474&6.793&7.165&7.589\\
6.142&6.084&6.849&7.927&7.800&8.447&8.842\\
7.476&7.070&7.603&8.183&8.531&9.190&9.460\\
6.758&7.108&8.385&8.746&9.120&9.800&10.60\\
7.934&7.924&8.75&9.911&10.24&10.81&11.68\\
9.320&9.950&12.03&12.85&13.59&14.09&14.80\\
\hline
9.090&9.860&11.76&12.57&13.32&13.93&14.52\\
12.33&14.91&17.69&18.38&18.73&19.36&20.00\\    
--&--&23.16&26.04&27.82&26.21&27.53\\
--&--&--&41.30&44.18&47.89&56.76\\
--&--&--&29.17&32.28&35.55&38.95\\
\hline
\end{tabular}}

\newcommand{\TblQueryBN}{\begin{tabular}{| c H H c c c c |}
\hline
\multicolumn{7}{|c|}{\textbf{CTL$_{bn}$ [$\mu$s]}}\\
\hline
\HdrTheta\\
\hline
14.87&15.30&16.01&16.31&16.55&16.94&17.45\\
16.34&16.75&17.61&17.80&18.04&18.32&18.78\\
25.05&26.22&26.87&27.21&27.51&27.75&28.24\\
34.95&36.09&36.55&37.01&37.54&38.18&38.61\\
45.95&47.28&48.07&48.07&47.70&49.91&50.46\\
48.62&49.79&51.28&51.10&51.66&52.64&53.73\\
59.10&60.06&61.39&61.73&62.51&63.48&64.55\\
97.85&99.12&101.3&101.8&102.3&103.6&105.0\\
\hline
98.77&100.2&102&103.8&104.2&104.9&106.0\\
168.8&172&172.9&175.4&175.2&178.1&179.1\\
293.4&295.4&297.2&298.8&298.0&302.1&304.0\\
421.1&424.4&427.2&428.9&430.2&431.4&432.0\\
132.3&135.3&136.8&138.9&140.6&140.8&142.4\\
\hline
\end{tabular}}

\newcommand{\TblQueryEN}{\begin{tabular}{| c H H c c c c |}
\hline
\multicolumn{7}{|c|}{\textbf{CTL$_{en}$ [$\mu$s]}}\\
\hline
\HdrTheta\\
\hline
7.366&10.46&9.25&9.822&10.428&11.41&12.41\\
8.125&10.43&9.611&10.09&10.512&11.04&11.73\\
12.56&15.57&14.54&15.09&15.53&16.32&17.00\\
15.86&18.98&17.72&18.39&18.80&19.72&20.46\\
21.65&24.97&24.14&24.56&25.13&25.84&26.52\\
22.49&27.6&25.82&26.59&27.50&29.42&30.94\\
26.31&30.28&28.92&29.52&30.40&31.77&33.44\\
42.69&50.9&45.86&48.32&50.82&52.88&55.02\\
\hline
47.83&54.03&49.9&51.42&53.89&54.99&55.88\\
80.33&84.86&82.47&83.47&85.45&90.74&92.55\\
115.2&122&119.6&120.8&122.6&125.2&127.9\\
174.8&182.3&179.8&180.5&182.2&184.3&186.8\\
61.28&76.09&75.63&75.84&76.36&79.48&83.80\\
\hline
\end{tabular}}

\newcommand{\TblQueryEB}{\begin{tabular}{| c H H c c c c |}
\hline
\multicolumn{7}{|c|}{\textbf{CTL$_{eb}$ [$\mu$s]}}\\
\hline
\HdrTheta\\
\hline
6.599&10.29&9.007&9.534&10.04&10.47&11.20\\
7.211&9.665&9.062&9.654&9.702&10.21&10.72\\
11.75&14.45&13.75&14.15&14.42&14.85&15.21\\
14.34&17.98&16.71&17.33&17.96&18.03&18.94\\
20.35&23.63&22.75&23.23&23.68&23.79&24.46\\
19.59&26.14&24.25&25.22&26.21&27.61&28.68\\
24.26&29.48&28.06&28.80&29.50&30.43&31.34\\
37.89&48.13&45.83&46.86&48.19&49.41&51.52\\
\hline
37.60&47.88&45.32&46.46&48.02&49.59&51.41\\
68.80&84.06&82.00&83.28&84.90&84.88&88.35\\
108.3&126.7&124.2&125.7&126.7&127.9&128.1\\
--&--&--&207.4&208.9&190.7&206.5\\
--&--&--&81.64&82.64&79.19&86.59\\
\hline
\end{tabular}}

\newcommand{\TblBatchBN}{\begin{tabular}{| c H H c c c c |}
\hline
\multicolumn{7}{|c|}{\textbf{CTL$_{bn}^b$ [$\mu$s]}}\\
\hline
\HdrTheta\\
\hline
5.448&4.604&4.954&4.767&5.088&5.994&6.483\\
5.133&4.357&5.357&4.972&5.669&5.949&5.945\\
7.641&5.939&6.932&6.997&7.297&7.549&7.086\\
8.848&7.558&7.926&8.232&8.474&8.820&8.407\\
9.566&7.954&8.266&9.292&9.543&9.897&10.25\\
10.86&9.456&9.2&10.35&10.78&11.39&12.09\\
11.17&10.27&10.67&10.98&10.49&11.83&11.38\\
14.89&12.41&12.73&13.04&13.57&14.23&15.03\\
\hline
15.04&13.5&13.79&14.08&14.37&14.96&15.72\\
20.51&17.02&17.19&17.46&17.85&18.45&19.19\\
29.86&26.51&26.96&27.48&28.06&28.91&30.04\\
36.77&31.57&32.11&32.20&32.47&33.45&34.24\\
23.44&22.94&23.74&24.45&25.71&26.28&27.19\\
\hline
\end{tabular}}

\newcommand{\TblBatchEN}{\begin{tabular}{| c H H c c c c |}
\hline
\multicolumn{7}{|c|}{\textbf{CTL$_{en}^b$ [$\mu$s]}}\\
\hline
\HdrTheta\\
\hline
5.213&5.042&5.728&6.304&6.894&7.792&8.696\\
5.529&5.17&5.673&6.027&6.427&7.028&7.523\\
7.174&6.339&7.17&7.518&8.018&8.241&8.933\\
7.965&7.079&7.656&8.039&8.529&9.228&9.901\\
9.461&8.568&9.178&9.370&9.957&10.66&11.30\\
9.736&9.199&10.3&10.49&11.27&12.55&14.25\\
10.85&9.29&10.17&10.68&11.38&12.43&13.40\\
14.39&13.46&14.42&15.24&16.29&17.39&18.99\\
\hline
13.82&13.08&13.3&13.98&14.96&16.47&17.90\\
17.82&16.15&15.88&16.39&17.37&18.74&20.05\\
23.44&21.49&22.33&23.45&24.97&27.31&29.80\\
28.84&25.63&26.00&26.97&28.24&30.50&32.80\\
19.95&23.30&25.46&26.56&28.49&32.84&37.67\\
\hline
\end{tabular}}

\newcommand{\TblBatchEB}{\begin{tabular}{| c H H c c c c |}
\hline
\multicolumn{7}{|c|}{\textbf{CTL$_{eb}^b$ [$\mu$s]}}\\
\hline
\HdrTheta\\
\hline
4.704&4.517&5.237&5.624&6.009&6.961&7.706\\
5.204&4.661&5.362&5.699&6.047&6.527&7.005\\
6.709&5.959&6.318&6.949&7.068&7.533&8.338\\
7.200&6.763&7.335&7.769&8.180&8.667&9.318\\
8.713&7.889&8.45&8.840&9.189&9.956&10.54\\
8.148&8.378&9.204&9.565&10.05&11.03&12.14\\
9.517&9.188&9.856&10.165&10.69&11.51&12.57\\
11.73&12.33&13.48&14.13&14.42&15.05&16.13\\
\hline
11.17&11.77&12.85&13.32&13.82&14.26&15.32\\
14.72&15.77&15.87&17.12&17.61&18.35&19.03\\
15.28&20.47&22.53&23.76&24.69&26.39&28.06\\
--&--&--&33.13&34.31&35.83&33.81\\
--&--&--&27.53&29.84&31.06&33.18\\
\hline
\end{tabular}}

\newcommand{\TblSizeDistance}{\begin{tabular}{| r H H r r r r | r | r |}
\hline
\multicolumn{9}{|c|}{\textbf{Distance-only Labels [MB]}}\\
\hline
\HdrTheta & \textbf{CCH} & \textbf{H2H}\\
\hline
147&114&69&57&49&40&36 & 16&298\\
119&93&55&46&40&35&32 & 14&288\\
194&149&83&68&57&49&45 & 18&606\\
473&367&213&174&148&125&116 & 48&1329\\
501&385&214&175&149&127&117 & 48&1570\\
1167&879&464&364&295&234&206 & 76&3196\\
1141&866&466&372&307&252&226 & 86&3438\\
2992&2210&1083&826&649&499&427 & 139&7979\\
\hline
3153&2344&1166&893&704&547&474 & 162&10920\\
5218&3885&1946&1501&1192&933&812 & 276&20336\\
20856&15212&7029&5190&3935&2886&2388 & 657&72991\\
37556&27387&12615&9258&6958&5036&4135 & 1101&\oom{c|}\\
38113&27654&12894&9346&6882&4742&3720 & 869&\oom{c|}\\
\hline
\end{tabular}}

\newcommand{\TblSizeBB}{\begin{tabular}{| r H H r r r r |}
\hline
\multicolumn{7}{|c|}{\textbf{CTL$_{bb}$ [MB]}}\\
\hline
\HdrTheta\\
\hline
268&202&112&89&72&56&47\\
214&161&86&69&56&46&40\\
356&265&134&103&82&65&57\\
859&647&339&262&208&164&144\\
916&682&340&263&210&167&147\\
2201&1625&795&596&456&335&278\\
2126&1576&775&588&458&347&296\\
5736&4172&1918&1405&1051&751&607\\
\hline
6007&4388&2033&1486&1108&795&649\\
9913&7245&3367&2477&1859&1341&1100\\
40471&29183&12817&9138&6630&4532&3536\\
\oom{|c}&--&23115&16401&11800&7956&6154\\
\oom{|c}&--&24209&17113&12185&7905&5861\\
\hline
\end{tabular}}

\newcommand{\TblSizeCCH}{\begin{tabular}{| r |}
\hline
\multirow{2}{*}{\textbf{CCH}}\\\\
\hline
21\\18\\24\\61\\61\\98\\110\\179\\
\hline
207\\352\\841\\1407\\1112\\
\hline
\end{tabular}}

\newcommand{\TblSizeEE}{\begin{tabular}{| r H H r r r r |}
\hline
\multicolumn{7}{|c|}{\textbf{CTL$_{ee}$ [MB]}}\\
\hline
\HdrTheta\\
\hline
1114&819&412&308&232&159&122\\
869&634&297&217&161&114&91\\
1477&1067&477&341&247&168&131\\
3534&2580&1195&848&608&409&320\\
3785&2734&1194&848&611&416&327\\
9397&6803&3068&2173&1544&1001&743\\
8973&6496&2892&2052&1466&970&736\\
24873&17832&7693&5384&3791&2441&1792\\
\hline
25899&18614&8017&5554&3855&2443&1788\\
42627&30623&13172&9168&6387&4055&2970\\
\oom{|c}&--&--&36432&25142&15702&11221\\
\oom{|c}&--&--&65812&45107&27808&19700\\
\oom{|c}&--&--&71016&48843&29582&20381\\
\hline
\end{tabular}}

\newcommand{\TblSizeHH}{\begin{tabular}{| r|}
\hline
\multirow{2}{*}{\textbf{H2H}}\\\\
\hline
336\\327\\649\\1309\\1833\\3211\\4128\\9510\\
\hline
10714\\20367\\84848\\--\\--\\
\hline
\end{tabular}}

\newcommand{\TblSizeBN}{\begin{tabular}{| r H H r r r r |}
\hline
\multicolumn{7}{|c|}{\textbf{CTL$_{bn}$ [MB]}}\\
\hline
\HdrTheta\\
\hline
151&118&73&61&53&44&40\\
123&97&59&50&44&39&36\\
199&154&88&73&62&54&50\\
486&380&226&187&161&138&129\\
514&398&227&188&162&140&130\\
1188&900&485&385&316&255&227\\
1164&889&489&395&330&275&249\\
3030&2248&1121&864&687&537&465\\
\hline
3196&2387&1209&936&747&590&517\\
5291&3958&2019&1574&1265&1006&885\\
21032&15388&7205&5366&4111&3062&2564\\
37850&27681&12909&9552&7252&5330&4429\\
38347&27888&13128&9580&7116&4976&3954\\
\hline
\end{tabular}}

\newcommand{\TblSizeEN}{\begin{tabular}{| r H H r r r r |}
\hline
\multicolumn{7}{|c|}{\textbf{CTL$_{en}$ [MB]}}\\
\hline
\HdrTheta\\
\hline
175&142&97&85&77&68&64\\
143&117&79&70&64&59&56\\
224&179&113&98&87&79&75\\
550&444&290&251&225&202&193\\
576&460&289&250&224&202&192\\
1294&1006&591&491&422&361&333\\
1279&1004&604&510&445&390&364\\
3222&2440&1313&1056&879&729&657\\
\hline
3412&2603&1425&1152&963&806&733\\
5657&4324&2385&1940&1631&1372&1251\\
21915&16271&8088&6249&4994&3945&3447\\
39321&29152&14380&11023&8723&6801&5900\\
39520&29061&14301&10753&8289&6149&5127\\
\hline
\end{tabular}}

\newcommand{\TblSizeEB}{\begin{tabular}{| r H H r r r r |}
\hline
\multicolumn{7}{|c|}{\textbf{CTL$_{eb}$ [MB]}}\\
\hline
\HdrTheta\\
\hline
312&246&156&133&116&100&91\\
251&198&123&106&93&83&77\\
403&312&181&150&129&112&104\\
981&769&461&384&330&286&266\\
1035&801&459&382&329&286&266\\
2401&1825&995&796&656&535&478\\
2345&1795&994&807&677&566&515\\
6099&4535&2281&1768&1414&1114&970\\
\hline
6418&4799&2444&1897&1519&1206&1060\\
10611&7943&4065&3175&2557&2039&1798\\
42151&30863&14497&10818&8310&6212&5216\\
\oom{|c}&--&--&19204&14603&10759&8957\\
\oom{|c}&--&--&19339&14411&10131&8087\\
\hline
\end{tabular}}

\begin{table*}[t]
 \centering
 \caption{Comparing customization time of CTL with CCH and query time for distance queries with CCH and H2H.}
 \label{table:distance-only}
 \setlength{\tabcolsep}{3pt} 
 \resizebox{\textwidth}{!}{%
    \TblNet\,\TblCustomization\,\TblQueryDistance
 }
\end{table*}

\begin{table*}[ht]
 \centering
 \caption{Comparing path query times for basic path unpacking in sequential and batch settings with CCH and H2H.}
 \label{table:path_query}
 \setlength{\tabcolsep}{3pt} 
 \resizebox{\textwidth}{!}{%
    \TblNet\,\TblQueryBB\,\TblQueryCCH\,\TblQueryHH\,\TblBatchBB\,\TblBatchCCH
 }
\end{table*}

\begin{table*}[ht]
 \centering
 \caption{Path query times for extended path unpacking with $k = 6$ in sequential and batch settings.}
 \label{table:extended_path_query}
 \shrinkbox{\textwidth}{!}{
    \TblNet\,\TblQueryEE\,\TblBatchEE
 }
\end{table*}

\begin{table*}[ht]
 \centering
 \caption{Path query times for basic/none path unpacking in sequential and batch settings.}
 \label{table:path_query_none_basic}
 \shrinkbox{\textwidth}{!}{%
    \TblNet\,\TblQueryBN\,\TblBatchBN
}
\end{table*}

\begin{table*}[ht]
 \centering
 \caption{Path query times for extended/none path unpacking in sequential and batch settings.}
 \label{table:path_query_none_extended}
 \shrinkbox{\textwidth}{!}{%
    \TblNet\,\TblQueryEN\,\TblBatchEN
}
\end{table*}

\begin{table*}[ht]
 \centering
 \caption{Path query times for extended/basic path unpacking in sequential and batch settings.}
 \label{table:path_query_basic_extended}
 \shrinkbox{\textwidth}{!}{%
    \TblNet\,\TblQueryEB\,\TblBatchEB
}
\end{table*}

\begin{table*}[ht]
 \centering
 \caption{Average path overlap across queries, exploited by CTL batch processing.}
 \label{table:overlap_percentage}
 \resizebox{\textwidth}{!}{%
 \begin{tabular}{| l | c c c c c c c c | c c c c c |}  \hline
    Batch Size&NY&BAY&COL&FLA&NW&NE&CAL&LKS&EUS&WUS&CUS&USA&EUR\\ \hline
    100&28\%&36\%&28\%&24\%&36\%&34\%&30\%&44\%&27\%&21\%&16\%&25\%&15\%\\
    1000&56\%&66\%&62\%&58\%&63\%&57\%&60\%&70\%&58\%&54\%&46\%&50\%&42\%\\
    10000&79\%&86\%&84\%&81\%&84\%&79\%&82\%&86\%&79\%&79\%&72\%&74\%&69\%\\
    100000&92\%&95\%&94\%&93\%&94\%&92\%&93\%&94\%&92\%&92\%&88\%&89\%&85\%\\
    1000000&98\%&99\%&98\%&98\%&98\%&97\%&98\%&98\%&97\%&97\%&95\%&96\%&93\%\\\hline
 \end{tabular}}
\end{table*}

\subsection{Performance Comparison}
In this section, we compare the performance of our method \om~with CCH across four metrics: \emph{query time (QT)} -- the time to retrieve the optimal cost and route; \emph{customization time (CT)} -- the time to customize the preprocessed data structures to a given metric; \emph{labeling size (LS)} -- the memory consumed by the precomputed data structures; and \emph{preprocessing time (PT)} -- the time to construct metric-independent data structures.
Against H2H we compare on query time, labeling size and preprocessing time. Columns named $\theta=x$ indicate CTL with threshold parameter $x$. Note that for $\theta=0$ results for some of the largest datasets are missing, as the labelings exceed memory limits.
This underscores the benefit of having a flexible tradeoffs between performance and resource requirements.

\subsubsection{Query Time}
After customizing each weight metric, we evaluate the performance of path queries using three settings: random queries, queries with varying lengths, and queries with varying batch sizes.
Tables~\ref{table:path_query} to \ref{table:path_query_basic_extended} use a fixed batch size of $10^6$, while in practice batch size would be determined by query load, which can vary over time.

\vspace{0.1cm}
\noindent\emph{Random Queries.}
Table~\ref{table:distance-only} shows times for distance queries. While our focus is on path queries, times for distance queries provide a useful baseline to judge how much time is spent on each phase of query evaluation.

Tables~\ref{table:path_query} to~\ref{table:path_query_basic_extended} report query times for our unpacking algorithms under sequential
(CTL$_{bn}$, \dots, CTL$_{ee}$)
and batch settings
(CTL$_{bn}^b$, \dots, CTL$_{ee}^b$).
The first thing to notice is that compared to distance queries, running times increase significantly, by more than two orders of magnitude in some cases. We also find that the impact of our threshold parameter $\theta$ on query times is much more subdued.
These effects occur as labels only speed up the first two phases, hub identification and endpoint chain construction, causing endpoint chain expansion to become the main bottleneck. This bottleneck is then mitigated through batch processing.

Across all networks, our method consistently outperforms CCH and H2H. In the sequential setting, \(\text{CTL}_{bb}\) achieves between 1.5 and 2.5 times improvements on the US networks.
For the EUR network CCH requires 524 $\mu$s, compared to CTL$_{bb}$ which requires only $\sim$130 $\mu$s, a $\times$4 speedup.
The two largest networks cause H2H and CTL$_{bb}$ with $\theta=0$ to run out of memory.
The difference for the European dataset is so pronounced as shortest paths tend to contain fewer edges (3-4 times fewer than USA).
Thus a smaller fraction of processing time is spent on path unpacking, and any speedup to hub identification and endpoint-chain construction has a greater impact on total query time.
Batch processing increases the gap between CTL and CCH further, up to a factor of 13-14 for EUR, as the resulting speedup to endpoint chain expansion makes any speedup to the first two steps more significant overall.

\begin{figure*}[!htbp]
    \centering
    \includegraphics[width=\textwidth]{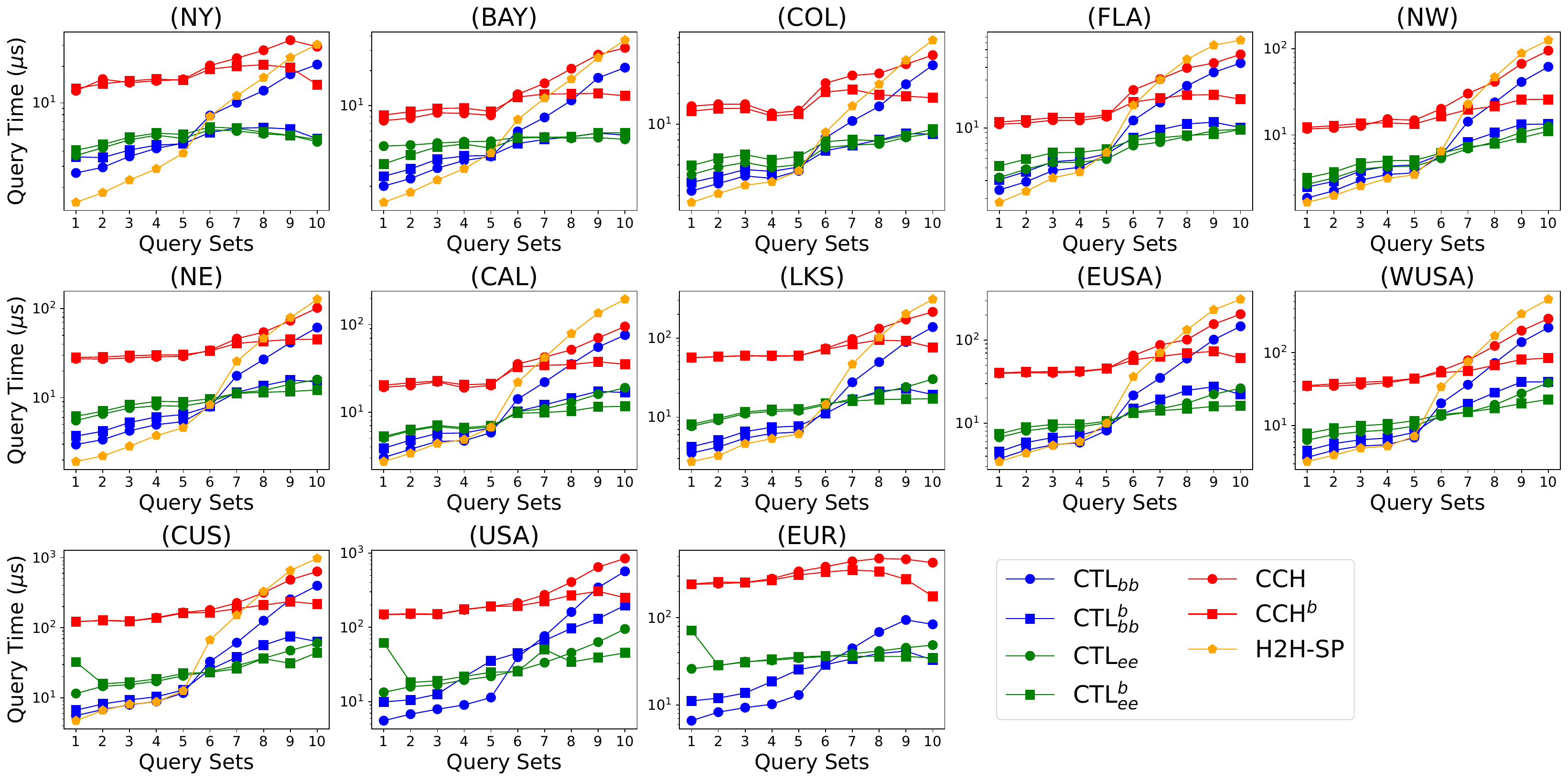}
    \caption{Query times across query sets of varying distances with $k = 6$ for extended path unpacking.}
    \label{fig:varying_distance}
\end{figure*}

\begin{figure*}[!htbp]
    \centering
    \includegraphics[width=\textwidth]{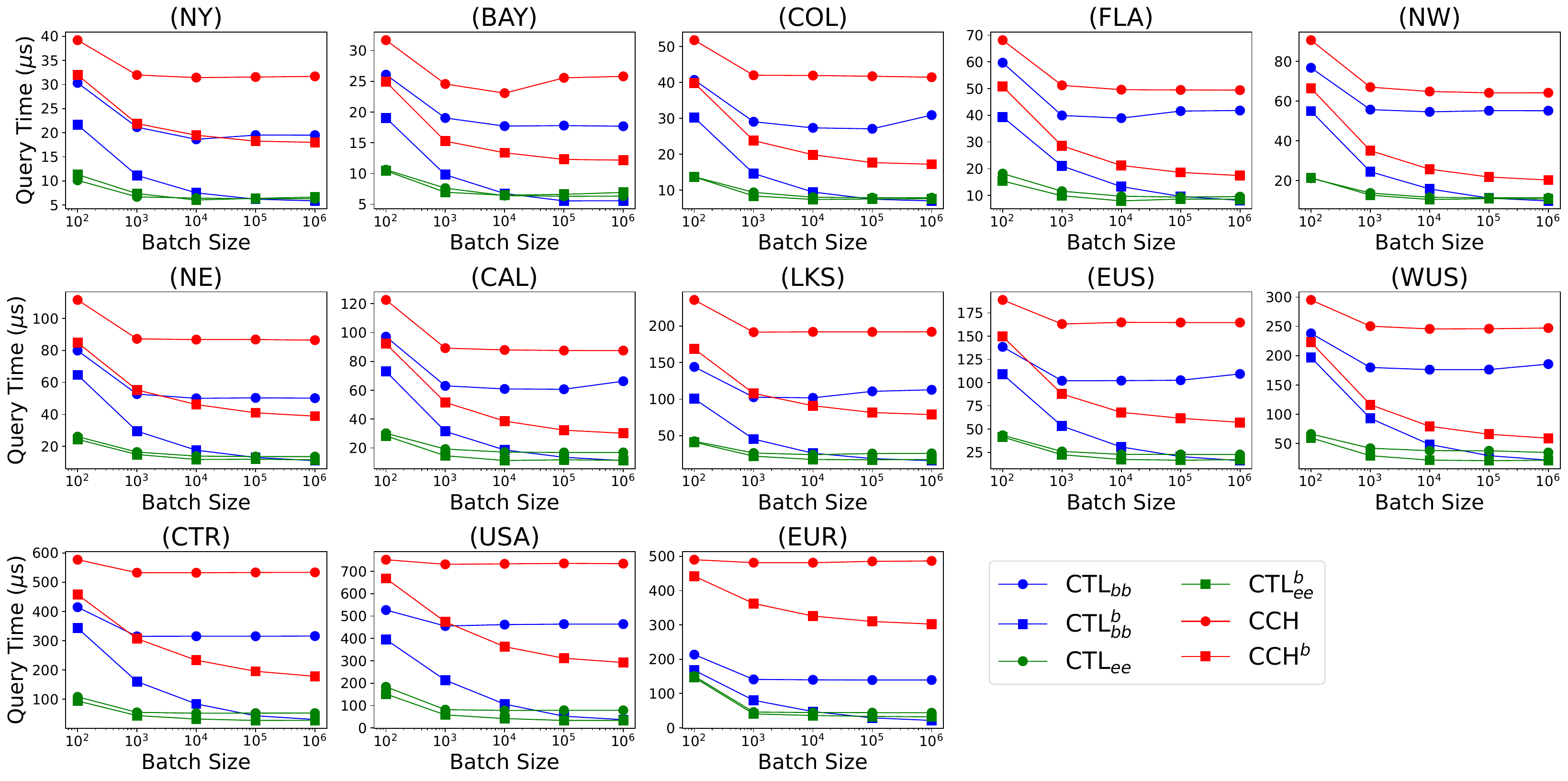}
    \caption{Query performance under varying batch sizes with $k = 6$ for extended path unpacking.}
    \label{fig:varying_batch_sizes}
\end{figure*}

The different CTL variants shown in Figure~\ref{fig:variants} behave largely as expected, with query times reducing as the information stored in shortcuts and labels increases.
In particular, extended path unpacking provides a major performance boost, with CTL$_{ee}$ achieving query times up to 5 times faster than CTL$_{bb}$.
Yet we also observe that batch processing combines poorly with CTL$_{ee}$, occasionally even increasing query times.
This happens as extended path unpacking and batch processing both target the same potential bottleneck.

However, the sheer scale of speed improvement from batch processing for variants other than CTL$_{ee}$ may be surprising, exceeding a full order of magnitude for CTL$_{bn}$ and CTL$_{bb}$ on the larger datasets, and reaching up to factor 6 for CTL$_{en}$ and CTL$_{eb}$.
The reason for this becomes evident from Table~\ref{table:overlap_percentage}, which shows percentages of overlap between query paths from hubs to endpoints, which CTL exploits to reduce endpoint chain unpacking. We observe that for sufficiently large batch sizes, overlap becomes so large that most of the endpoint unpacking work can be skipped.

\vspace{0.1cm}
\noindent\emph{Queries with varying lengths.}  
Figure~\ref{fig:varying_distance} presents results on query sets with varying lengths using $\theta=20$ for basic and extended path unpacking and $k=6$ for the latter. Across all networks, our proposed methods maintain a clear advantage over CCH, CCH\(^b\) and H2H.
While H2H is fastest amongst all methods for short-range queries, its performance quickly declines for longer ranges.
The impact of batch processing becomes most pronounced for long-range queries, where the relative cost of endpoint chain unpacking is largest, due to shortcuts being longest at the top of our hierarchy. For short-range queries batch processing provides little benefit, due to shorter paths with little overlap. This suggest that answering short-range queries sequentially (to avoid overheads and latency) and long-range queries in batches may be a viable strategy.
Note however that for clustered short-range queries, e.g. within one city embedded in a larger network, path overlap and batch processing efficiency increase again.

\vspace{0.1cm}
\noindent\emph{Queries with varying batch sizes.}
We process 1 million random queries in batches of size $\{10^2,\ldots,10^6\}$, using $\theta=20$ for all path variants and $k=6$ for extended path unpacking. Figure~\ref{fig:varying_batch_sizes} shows that batch query processing consistently lowers per-query cost for CTL$_{bb}$ and CCH, with the effect most pronounced on large networks, but has little impact on CTL$_{ee}$, which already has excellent unpacking performance. As a result, CTL$_{bb}$ can rival CTL$_{ee}$ performance once batch size becomes sufficiently large, at around $10^5$.
H2H does not provide specialized batch-processing. While CCH benefits from batch processing just as much as CTL$_{bb}$ in absolute terms, the \emph{relative} benefit is smaller due to the higher cost for hub identification, causing the relative performance gap to widen.

\begin{figure*}
    \centering
    \includegraphics[width=\textwidth]{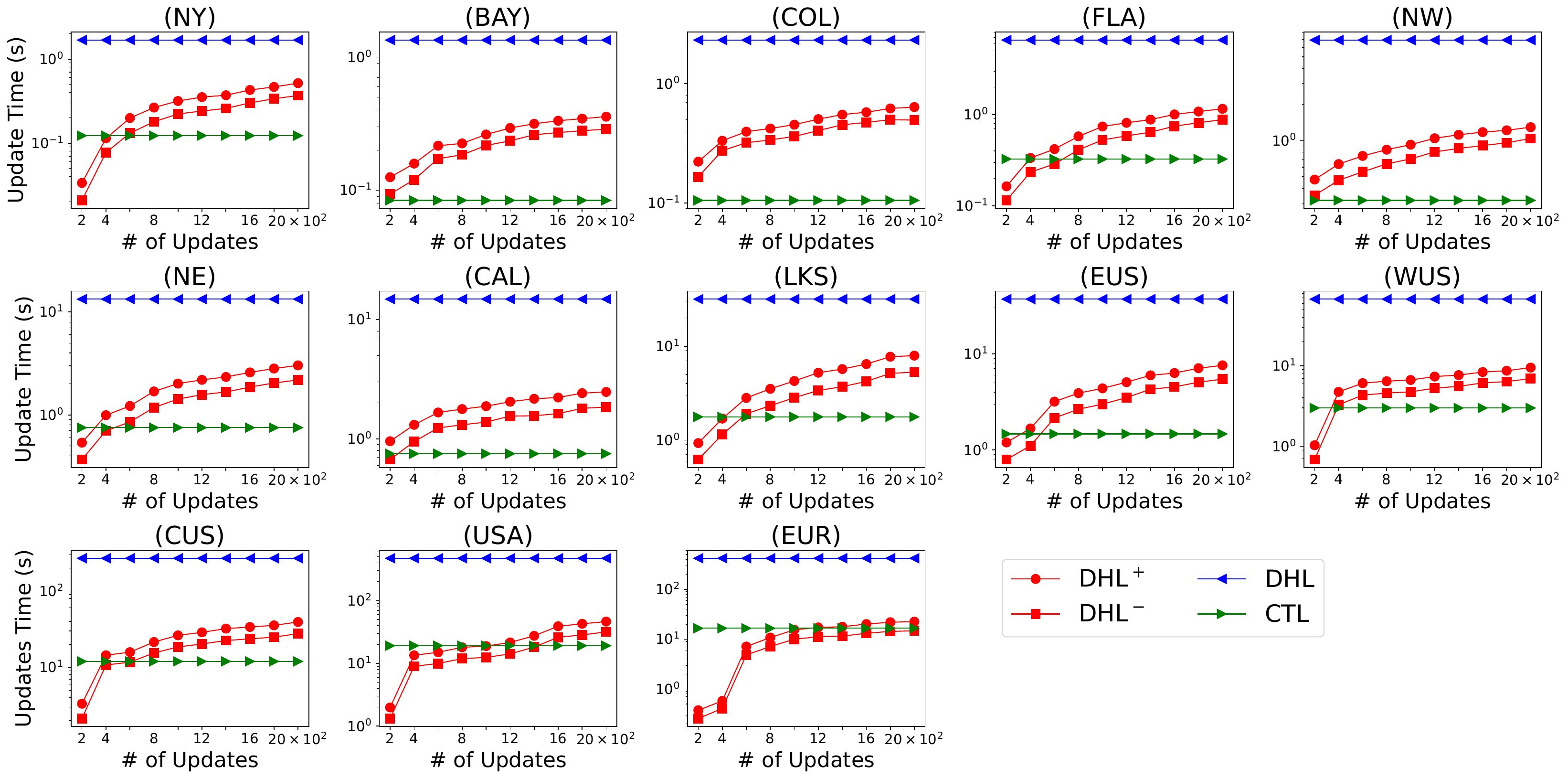}
    \caption{Update time analysis of DHL with varying numbers of updates for both edge cost increases and decreases, compared to the reconstruction time of DHL and the customization time of \om.}
    \label{fig:update_time}
\end{figure*}

\subsubsection{Customization Time}
Table~\ref{table:distance-only} shows that \om{} requires more customization time than CCH, as it must additionally customize tree labelings. On small and mid-sized networks \om{} at $\theta=0$ is about 2-3 times slower than CCH, which increases to 3-4 times on large networks. However, increasing $\theta$ substantially reduces this gap: customization times drop by a factor of two to three between $\theta=0$ and $\theta=100$, bringing \om{} much closer to CCH (for example, $19.1$s compared to $12.4$s on USA).

\begin{table*}[ht]
 \centering
 \caption{Labeling size supporting distance and path queries.}
 \label{table:labeling-sizes}
 \setlength{\tabcolsep}{3pt} 
 \resizebox{\textwidth}{!}{%
    \TblNet\,\TblSizeDistance\,\TblSizeBB\,\TblSizeCCH\,\TblSizeHH
 }
\end{table*}

\begin{table*}[ht]
 \centering
 \caption{Labeling size for CTL variants with $k = 6$ for extended path unpacking.}
 \label{table:labeling-sizes-extended}
 \shrinkbox{\textwidth}{!}{%
    \TblNet\,\TblSizeEB\,\TblSizeEE
 }
\end{table*}

\begin{table*}[ht]
 \centering
 \caption{Labeling size for CTL variants with no path information in $L$.}
 \label{table:labeling-sizes-none}
 \shrinkbox{\textwidth}{!}{%
    \TblNet\,\TblSizeBN\,\TblSizeEN
 }
\end{table*}

\subsubsection{Labeling Size}
The labeling size results in Tables~\ref{table:labeling-sizes} to~\ref{table:labeling-sizes-none} follow the same trend. At $\theta=0$, \om{} produces very large labelings, sometimes more than an order of magnitude larger than CCH, since nearly all search information is precomputed. As $\theta$ increases, labeling size decreases drastically though, shrinking by a factor of about six on average between $\theta=0$ and $\theta=100$. On the largest networks, such as USA and EUR, the labeling size drops from more than $37$GB at $\theta=0$ to about $4$GB at $\theta=100$, while CCH consistently stays below $1$GB.
Unsurprisingly, labeling sizes increase even further for extended path unpacking, especially for CTL$_{ee}$ which requires around 4 times more space than CTL$_{bb}$. While labeling size of H2H exceed that of \om\ by factor 2-4 when catering for distance queries only, support for path computation increases its size only slightly, situating it between CTL$_{eb}$ and CTL$_{ee}$ with $\theta=0$.

\begin{figure}
    \centering
    \includegraphics[width=0.67\textwidth]{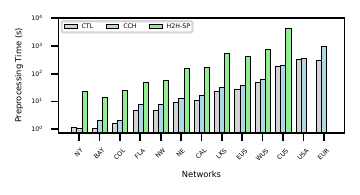}
    \caption{Preprocessing time of \om\ and CCH, and construction time of H2H.}
    \label{fig:preprocessing}
\end{figure}

\subsubsection{Preprocessing Time}
In Figure~\ref{fig:preprocessing} we compare the preprocessing time of CTL against CCH. The preprocessing time of \om~represents the total time spent on constructing the tree hierarchy \(H_G\) and the shortcut graph \(S_G\). Before constructing \(H_G\), we contract the road network by repeatedly removing degree-one vertices, following a similar approach as described in~\cite{farhan2023hierarchical}. As CCH still needs to construct the hierarchy to obtain a good vertex ordering, preprocessing times are virtually identical otherwise -- the only extra work performed by CTL is allocation of label memory. In particular this also means that the choice of $\theta$ has no impact on preprocessing time.
For H2H there is no clean separation between preprocessing and customization, and construction is more than an order of magnitude slower than preprocessing for CTL and CCH (and more than 2 orders slower than customization, cf. Table~\ref{table:distance-only}).

\subsubsection{Tradeoffs}
As we have seen, there are different tradeoffs to be had between labeling size and customization time on the one hand, and query time on the other.
What's more, these tradeoffs can be achieved in two different ways: choice of threshold parameter $\theta$, and choice of path information stored for shortcuts and label entries.
This raises the question which combinations of these options result in the best overall tradeoffs.

To answer this, we plot pairs of (labeling size, query time) achievable by the different configurations considered in this section in Figure~\ref{fig:tradeoffs}.
Here we ignore customization time for the sake of simplicity, but note that it is affected by $\theta$ in the same way as labeling size.

We observe that CTL$_{ee}$ with high $\theta$ values -- as a rule-of-thumb we recommend $\theta\approx\sqrt{|V|/1000}$ -- provides the best possible trade-offs where a direct comparison is possible.
However, the range of trade-offs CTL$_{ee}$ provides is limited to relatively large labeling sizes, at least 10GB for USA and EUR.
If smaller labeling sizes are desired, CTL$_{en}$ and CTL$_{eb}$ offer the next best options.
Essentially we find that extended path unpacking is critical for good performance, but using it for shortcuts only (and not for label entries) may be sufficient.

\begin{figure*}[t]
    \centering
    \includegraphics[width=\textwidth]{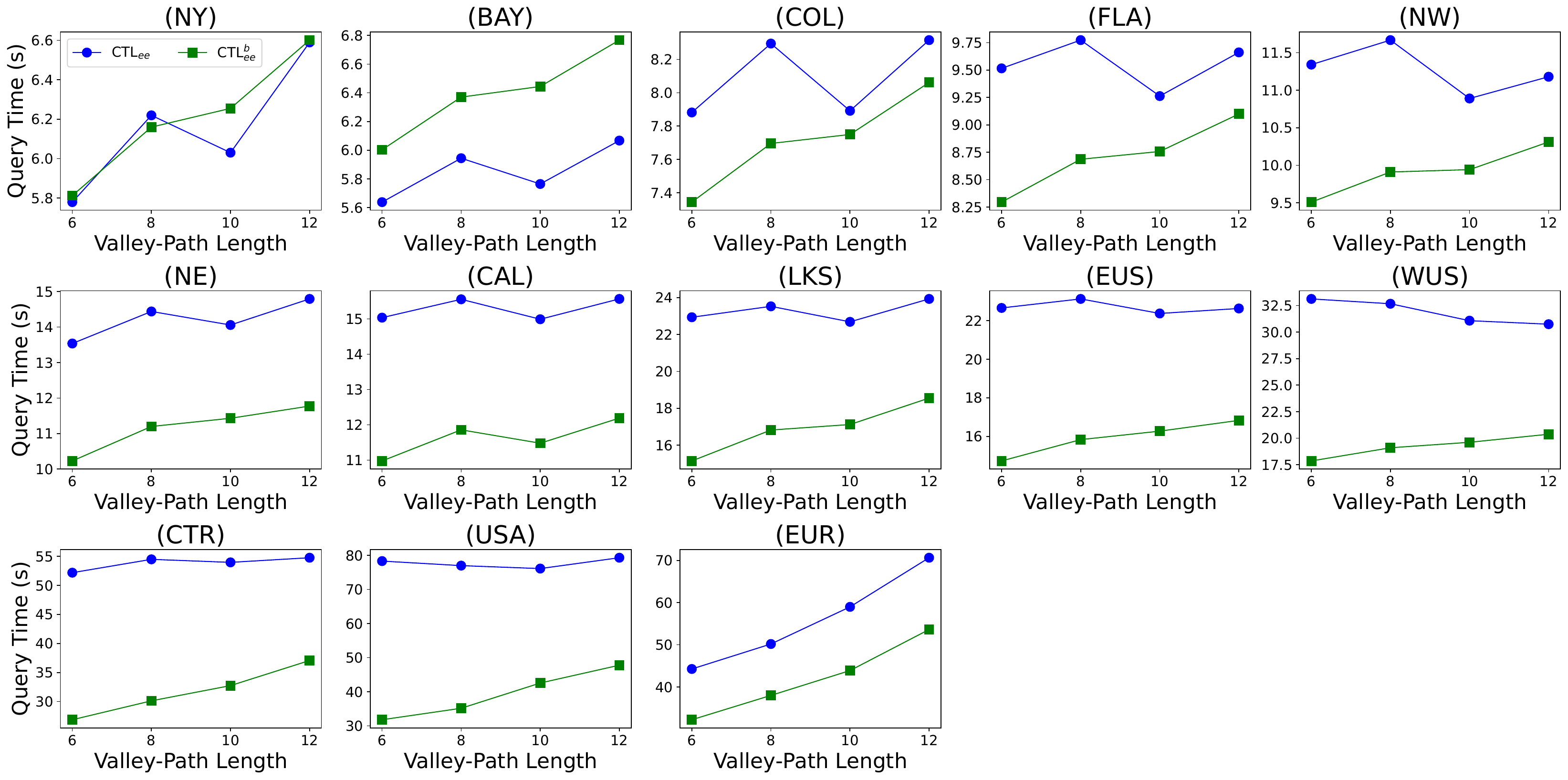}
    \caption{Query performance of extended path unpacking variant of CTL using varying valley-path lengths $k$.}
    \label{fig:varying_extendedpath_length}
\end{figure*}
\begin{figure}[t]
    \centering
    \includegraphics[width=0.6\textwidth]{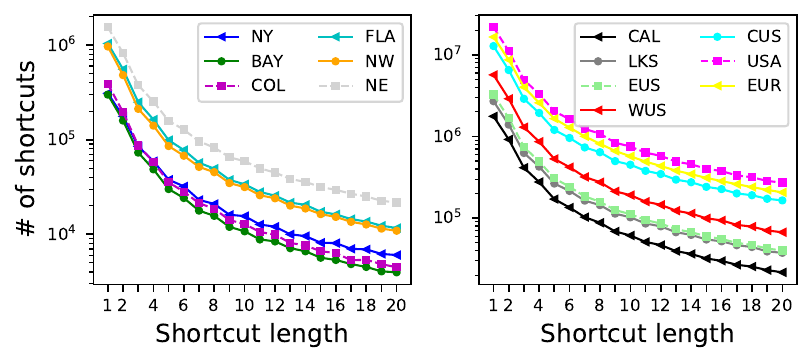}
    \caption{Shortcuts distribution.}
    \label{fig:shortcut_dist}
\end{figure}

Note also that Figure~\ref{fig:tradeoffs} shows only the sequential scenario.
For use cases where batch processing is viable, CTL$_{en}$ or CTL$_{eb}$ should almost always be preferable over CTL$_{ee}$, as batch processing provides very limited performance gains for the latter, as seen in Table~\ref{table:extended_path_query} and Figures~\ref{fig:varying_distance} and~\ref{fig:varying_batch_sizes}.

\subsubsection{Partial Customization}
In Figure~\ref{fig:update_time}, we compare our method \om~with the state-of-the-art partial customizable method DHL \cite{farhan2025dual} that incrementally maintains its labels to reflect small numbers of edge cost updates in road networks.
We compared the update time over batches of updates with sizes ranging from 200, 2,000 to the customization time of \om~with $\theta=20$ for the smaller networks and $\theta=100$ for the larger ones, as well as to DHL's preprocessing time (rebuilding from scratch). Figure~\ref{fig:update_time} shows that DHL's update time increases with the number of updates, surpassing \om's customization time at around 400-1000 updates on most networks, with the exception of EUR where around 2000 updates are required to break even. Yet 2000 updates on EUR means less than 0.005\% of edges changed their weight.
Overall, this indicates that \om~is preferable for applications where updates are frequent, even moderately so.
The reconstruction time of DHL is more than an order of magnitude slower than customization time of \om.

\begin{figure}
    \centering
    \includegraphics[width=0.7\textwidth]{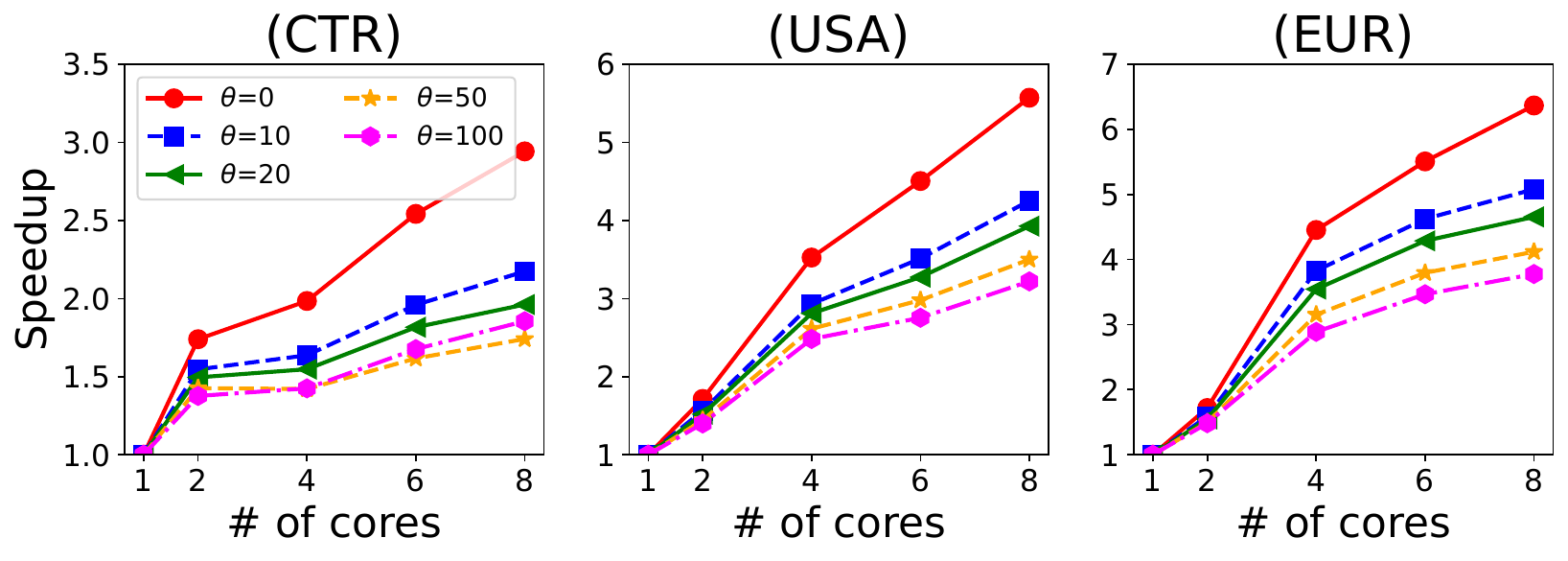}
    \caption{Customization performance with increasing the \# of cores on largest three networks.}
    \label{fig:speedup_scal}
\end{figure}

\subsection{Further Analysis}

We finish by investigating the maximal number $k$ of intermediate valley path nodes stored for extended path unpacking, and parallelization performance.

\subsubsection{Valley Path Length}

In the prior experiments we had set the maximal number of intermediate nodes stored for each shortcut to the largest value $k=6$ which did not increase storage space (see Figure~\ref{fig:path-data}).
However, as discussed in Section~\ref{subsec:extended-valley-unpacking}, it might be possible to improve performance further by increasing this number, offering yet another time-space tradeoff.
We tested this by recording query times for $k\in\{6,8,10,12\}$, with results shown in Figure~\ref{fig:varying_extendedpath_length}.
Here we find that query times tend to increase for larger $k$ values, rather than decrease as one might expect.
To explain this behavior, we examined how long valley paths actually are across the different networks.
As Figure~\ref{fig:shortcut_dist} shows, the vast majority of valley paths contains 7 edges or fewer, so a limit of $k=6$ intermediate nodes will still allow most of them to be stored directly.
Thus increasing $k$ further does not reduce the amount of unpacking work substantially, while the increased storage requirements impede caching.

\subsubsection{Parallel Customization} 
Figure~\ref{fig:speedup_scal} shows the performance of \om~on the three largest networks, CUS, USA, and EUR, when varying the number of cores from 1 to 8. These networks are well-suited for parallel processing due to the large number of vertices at each tree hierarchy level, enabling efficient workload distribution across cores. \om~achieves near-linear performance improvement with increasing cores, demonstrating its scalability and effectiveness in leveraging parallel processing for large-scale networks.
\section{Conclusion}\label{section:conclusion}
In this paper, we study the problem of route optimization to enhance route planning in real-time and dynamic road networks. We build on the state-of-the-art Customizable Tree Labeling (CTL) framework, which supports efficient metric-independent preprocessing and metric-dependent customization of tree labelings to support fast queries. To address its limitations, we systematically explore how path information is stored in shortcut graphs and path arrays, and investigate multiple algorithmic variants that offer different trade-offs between memory usage and query speed. We also introduce a batch processing approach that reuses path information across multiple queries, thereby reducing redundant computation and enhancing scalability. Experiments on thirteen large real-world road networks show that our algorithms substantially outperform existing methods in computing optimal routes, demonstrating their practical effectiveness in dynamic and real-time routing scenarios.

\bibliographystyle{ACM-Reference-Format}
\bibliography{sample}

@inproceedings{hu2025reproducibility,
  title={Reproducibility Report for ACM SIGMOD 2024 Paper:'Hierarchical Cut Labelling-Scaling Up Distance Queries on Road Networks'},
  author={Hu, Yihao and Oliaro, Gabriele and Yang, Jiani and Farhan, Muhammad},
  booktitle={Reproducibility Reports of the 2024 International Conference on Management of Data},
  pages={8--9},
  year={2025}
}

@article{ahn2008effects,
  title={The effects of route choice decisions on vehicle energy consumption and emissions},
  author={Ahn, Kyoungho and Rakha, Hesham},
  journal={Transportation Research Part D: transport and environment},
  volume={13},
  number={3},
  pages={151--167},
  year={2008},
  publisher={Elsevier}
}

@article{farhan2025customization,
  title={Customization Meets 2-Hop Labeling: Efficient Routing in Road Networks},
  author={Farhan, Muhammad and Koehler, Henning and Wang, Qing and Wang, Jiawen and Laupichler, Moritz and Sanders, Peter},
  journal={Proceedings of the VLDB Endowment},
  volume={18},
  number={10},
  pages={3326--3338},
  year={2025},
  publisher={VLDB Endowment}
}

@article{farhan2025dual,
  title={Dual-Hierarchy Labelling: Scaling Up Distance Queries on Dynamic Road Networks},
  author={Farhan, Muhammad and Koehler, Henning and Wang, Qing},
  journal={Proceedings of the ACM on Management of Data},
  volume={3},
  number={1},
  pages={1--25},
  year={2025},
  publisher={ACM New York, NY, USA}
}

@article{koehler2025stable,
  title={Stable Tree Labelling for Accelerating Distance Queries on Dynamic Road Networks},
  author={Koehler, Henning and Farhan, Muhammad and Wang, Qing},
  journal={arXiv preprint arXiv:2501.17379},
  year={2025}
}

@article{blasius2025customizable,
  title={Customizable Contraction Hierarchies--A Survey},
  author={Bl{\"a}sius, Thomas and Buchhold, Valentin and Wagner, Dorothea and Zeitz, Tim and Z{\"u}ndorf, Michael},
  journal={arXiv preprint arXiv:2502.10519},
  year={2025}
}

@inproceedings{delling2011graph,
  title={Graph partitioning with natural cuts},
  author={Delling, Daniel and Goldberg, Andrew V and Razenshteyn, Ilya and Werneck, Renato F},
  booktitle={2011 IEEE International Parallel \& Distributed Processing Symposium},
  pages={1135--1146},
  year={2011},
  organization={IEEE}
}

@inproceedings{zhang2022relative,
  title={Relative Subboundedness of Contraction Hierarchy and Hierarchical 2-Hop Index in Dynamic Road Networks},
  author={Zhang, Yikai and Yu, Jeffrey Xu},
  booktitle={Proceedings of the 2022 International Conference on Management of Data},
  pages={1992--2005},
  year={2022}
}

@article{geisberger2012exact,
  title={Exact routing in large road networks using contraction hierarchies},
  author={Geisberger, Robert and Sanders, Peter and Schultes, Dominik and Vetter, Christian},
  journal={Transportation Science},
  volume={46},
  number={3},
  pages={388--404},
  year={2012},
  publisher={INFORMS}
}

@article{bast2016route,
  title={Route planning in transportation networks},
  author={Bast, Hannah and Delling, Daniel and Goldberg, Andrew and M{\"u}ller-Hannemann, Matthias and Pajor, Thomas and Sanders, Peter and Wagner, Dorothea and Werneck, Renato F},
  journal={Algorithm engineering: Selected results and surveys},
  pages={19--80},
  year={2016},
  publisher={Springer}
}

@article{delling2017customizable,
  title={Customizable route planning in road networks},
  author={Delling, Daniel and Goldberg, Andrew V and Pajor, Thomas and Werneck, Renato F},
  journal={Transportation Science},
  volume={51},
  number={2},
  pages={566--591},
  year={2017},
  publisher={INFORMS}
}

@article{dibbelt2016customizable,
  title={Customizable contraction hierarchies},
  author={Dibbelt, Julian and Strasser, Ben and Wagner, Dorothea},
  journal={Journal of Experimental Algorithmics (JEA)},
  volume={21},
  pages={1--49},
  year={2016},
  publisher={ACM New York, NY, USA}
}

@inproceedings{blum2022customizable,
    title={Customizable Hub Labeling: Properties and Algorithms},
    author={Blum, Johannes and Storandt, Sabine},
    booktitle={International Computing and Combinatorics Conference},
    pages={345--356},
    year={2022},
    organization={Springer}
}

@inproceedings{ouyang2018hierarchy,
  title={When hierarchy meets 2-hop-labeling: Efficient shortest distance queries on road networks},
  author={Ouyang, Dian and Qin, Lu and Chang, Lijun and Lin, Xuemin and Zhang, Ying and Zhu, Qing},
  booktitle={Proceedings of the ACM SIGMOD International Conference on Management of Data},
  pages={709--724},
  year={2018}
}

@article{ouyang2023hierarchy,
  author       = {Dian Ouyang and
                  Dong Wen and
                  Lu Qin and
                  Lijun Chang and
                  Xuemin Lin and
                  Ying Zhang},
  title        = {When hierarchy meets 2-hop-labeling: efficient shortest distance and
                  path queries on road networks},
  journal      = {{VLDB} J.},
  volume       = {32},
  number       = {6},
  pages        = {1263--1287},
  year         = {2023},
}

@inproceedings{chen2021p2h,
  title={P2h: Efficient distance querying on road networks by projected vertex separators},
  author={Chen, Zitong and Fu, Ada Wai-Chee and Jiang, Minhao and Lo, Eric and Zhang, Pengfei},
  booktitle={Proceedings of the ACM SIGMOD International Conference on Management of Data},
  pages={313--325},
  year={2021}
}

@book{tarjan1983data,
  title={Data structures and network algorithms},
  author={Tarjan, Robert Endre},
  year={1983},
  publisher={SIAM}
}

@inproceedings{goldberg2005computing,
  title={Computing the shortest path: {$A^*$} search meets graph theory},
  author={Goldberg, Andrew V and Harrelson, Chris},
  booktitle={SODA},
  volume={5},
  pages={156--165},
  year={2005}
}

@inproceedings{geisberger2008contraction,
  title={Contraction hierarchies: Faster and simpler hierarchical routing in road networks},
  author={Geisberger, Robert and Sanders, Peter and Schultes, Dominik and Delling, Daniel},
  booktitle={International workshop on experimental and efficient algorithms},
  pages={319--333},
  year={2008}
}

@inproceedings{akiba2013fast,
  title={Fast exact shortest-path distance queries on large networks by pruned landmark labeling},
  author={Akiba, Takuya and Iwata, Yoichi and Yoshida, Yuichi},
  booktitle={Proceedings of the ACM SIGMOD International Conference on Management of Data},
  pages={349--360},
  year={2013}
}

@inproceedings{akiba2014fast,
  title={Fast shortest-path distance queries on road networks by pruned highway labeling},
  author={Akiba, Takuya and Iwata, Yoichi and Kawarabayashi, Ken-ichi and Kawata, Yuki},
  booktitle={2014 Proceedings of the sixteenth workshop on algorithm engineering and experiments (ALENEX)},
  pages={147--154},
  year={2014}
}

@inproceedings{jin2012highway,
  title={A highway-centric labeling approach for answering distance queries on large sparse graphs},
  author={Jin, Ruoming and Ruan, Ning and Xiang, Yang and Lee, Victor},
  booktitle={Proceedings of the ACM SIGMOD International Conference on Management of Data},
  pages={445--456},
  year={2012}
}

@article{bast2006transit,
  title={Transit ultrafast shortest-path queries with linear-time preprocessing},
  author={Bast, Holger and Funke, Stefan and Matijevic, Domagoj},
  journal={9th DIMACS Implementation Challenge [1]},
  year={2006},
}

@inproceedings{bodlaender2006treewidth,
  title={Treewidth: characterizations, applications, and computations},
  author={Bodlaender, Hans L},
  booktitle={32nd International Workshop of Graph-Theoretic Concepts in Computer Science},
  pages={1--14},
  year={2006}
}

@inproceedings{10.1145/2463676.2465277,
    author = {Zhu, Andy Diwen and Ma, Hui and Xiao, Xiaokui and Luo, Siqiang and Tang, Youze and Zhou, Shuigeng},
    title = {Shortest Path and Distance Queries on Road Networks: Towards Bridging Theory and Practice},
    year = {2013},
    booktitle = {Proceedings of the ACM SIGMOD International Conference on Management of Data},
    pages = {857–868}
}

@phdthesis{Pohl1969BidirectionalAH,
author = {Pohl, Ira Sheldon},
title = {Bi-Directional and Heuristic Search in Path Problems},
year = {1969},
publisher = {Stanford University},
address = {Stanford, CA, USA},
note = {AAI7001588}
}

@inproceedings{abraham2011hub, 
    author = {Abraham, Ittai and Delling, Daniel and Goldberg, Andrew V. and Werneck, Renato F.}, 
    title = {A Hub-Based Labeling Algorithm for Shortest Paths in Road Networks}, 
    year = {2011}, 
    booktitle = {Proceedings of the 10th International Conference on Experimental Algorithms}, 
    pages = {230–241}, 
    numpages = {12}
}

@inproceedings{abraham2012hierarchical, 
    author = {Abraham, Ittai and Delling, Daniel and Goldberg, Andrew V. and Werneck, Renato F.}, 
    title = {Hierarchical Hub Labelings for Shortest Paths}, 
    year = {2012}, 
    booktitle = {Proceedings of the 20th Annual European Conference on Algorithms}, 
    pages = {24–35}, 
    numpages = {12} 
}

@article{hart1968formal,
  title={A formal basis for the heuristic determination of minimum cost paths},
  author={Hart, Peter E and Nilsson, Nils J and Raphael, Bertram},
  journal={IEEE transactions on Systems Science and Cybernetics},
  volume={4},
  number={2},
  pages={100--107},
  year={1968}
}

@inproceedings{10.1007/11561071_51,
    author = {Sanders, Peter and Schultes, Dominik},
    title = {Highway Hierarchies Hasten Exact Shortest Path Queries},
    year = {2005},
    booktitle = {Proceedings of the 13th Annual European Conference on Algorithms}
}

@inproceedings{sanders2006engineering,
  title={Engineering highway hierarchies},
  author={Sanders, Peter and Schultes, Dominik},
  booktitle={European Symposium on Algorithms},
  pages={804--816},
  year={2006},
  organization={Springer}
}

@article{maue2010goal,
  title={Goal-directed shortest-path queries using precomputed cluster distances},
  author={Maue, Jens and Sanders, Peter and Matijevic, Domagoj},
  journal={Journal of Experimental Algorithmics (JEA)},
  volume={14},
  pages={3--2},
  year={2010},
  publisher={ACM New York, NY, USA}
}

@article{jung2002efficient,
    title={An efficient path computation model for hierarchically structured topographical road maps},
    author={Jung, Sungwon and Pramanik, Sakti},
    journal={IEEE Transactions on Knowledge and Data Engineering},
    volume={14},
    number={5},
    pages={1029--1046},
    year={2002}
}

@inproceedings{arz2013transit,
    title={Transit node routing reconsidered},
    author={Arz, Julian and Luxen, Dennis and Sanders, Peter},
    booktitle={Proceedings of the 12th International Symposium of Experimental Algorithms},
    pages={55--66},
    year={2013}
}

@article{cohen2003reachability,
    title={Reachability and distance queries via 2-hop labels},
    author={Cohen, Edith and Halperin, Eran and Kaplan, Haim and Zwick, Uri},
    journal={SIAM Journal on Computing},
    volume={32},
    number={5},
    pages={1338--1355},
    year={2003}
}

@article{qiu2022efficient,
  title={Efficient shortest path counting on large road networks},
  author={Qiu, Yu-Xuan and Wen, Dong and Qin, Lu and Li, Wentao and Li, Rong-Hua and Ying, Zhang and others},
  journal={Proceedings of the VLDB Endowment},
  year={2022},
  publisher={Association for Computing Machinery (ACM)}
}

@misc{ptvplanung,
    title={Western europe dataset},
    author={PTV AG},
    year={},
    url={http://www.ptv.de}
}

@book{demetrescu2009shortest,
    title={The shortest path problem: Ninth DIMACS implementation challenge},
    author={Demetrescu, Camil and Goldberg, Andrew V and Johnson, David S},
    volume={74},
    year={2009},
    publisher={American Mathematical Soc.}
}

@inproceedings{sanders2005highway,
  title={Highway hierarchies hasten exact shortest path queries},
  author={Sanders, Peter and Schultes, Dominik},
  booktitle={Proceedings of the 13th annual European conference on Algorithms},
  pages={568--579},
  year={2005}
}

@inproceedings{farhan2023hierarchical,
  title={Hierarchical Cut Labelling--Scaling Up Distance Queries on Road Networks},
  author={Farhan, Muhammad and Koehler, Henning and Ohms, Robert and Wang, Qing},
  booktitle={Proceedings of the ACM SIGMOD International Conference on Management of Data},
  year={2024}
}

@article{ouyang2020efficient,
  title={Efficient shortest path index maintenance on dynamic road networks with theoretical guarantees},
  author={Ouyang, Dian and Yuan, Long and Qin, Lu and Chang, Lijun and Zhang, Ying and Lin, Xuemin},
  journal={Proceedings of the VLDB Endowment},
  volume={13},
  number={5},
  pages={602--615},
  year={2020}
}

@inproceedings{zhang2021dynamic,
  title={Dynamic hub labeling for road networks},
  author={Zhang, Mengxuan and Li, Lei and Hua, Wen and Mao, Rui and Chao, Pingfu and Zhou, Xiaofang},
  booktitle={IEEE 37th International Conference on Data Engineering (ICDE)},
  pages={336--347},
  year={2021}
}

@article{farhan2018highly,
  title={A highly scalable labelling approach for exact distance queries in complex networks},
  author={Farhan, Muhammad and Wang, Qing and Lin, Yu and Mckay, Brendan},
  journal={arXiv preprint arXiv:1812.02363},
  year={2018}
}

@article{george1973nested,
  title={Nested dissection of a regular finite element mesh},
  author={George, Alan},
  journal={SIAM journal on numerical analysis},
  volume={10},
  number={2},
  pages={345--363},
  year={1973},
  publisher={SIAM}
}

@InProceedings{luxen2011hierarchy,
author="Luxen, Dennis and Sanders, Peter",
editor="Pardalos, Panos M. and Rebennack, Steffen",
title="Hierarchy Decomposition for Faster User Equilibria on Road Networks",
booktitle="Experimental Algorithms",
year="2011",
publisher="Springer Berlin Heidelberg",
address="Berlin, Heidelberg",
pages="242--253"
}

@phdthesis{Schulz13a,
  author       = {Christian Schulz},
  title        = {High Quality Graph Partitioning},
  school       = {Karlsruhe Institute of Technology},
  year         = {2013},
  url          = {http://digbib.ubka.uni-karlsruhe.de/volltexte/1000035713},
  urn          = {urn:nbn:de:swb:90-357133},
  bibsource    = {dblp computer science bibliography, https://dblp.org}
}

\end{document}